\documentclass[lettersize,journal]{IEEEtran}
\usepackage{amsmath,amsfonts}
\usepackage[ruled]{algorithm2e}
\usepackage{array}
\usepackage[caption=false,font=normalsize,labelfont=sf,textfont=sf]{subfig}
\usepackage{textcomp}
\usepackage{stfloats}
\usepackage{url}
\usepackage{amssymb}
\usepackage{verbatim}
\usepackage{color}
\usepackage{soul}
\usepackage{graphicx}
\usepackage{cuted}
\usepackage{bm} 
\usepackage{mathrsfs}
\usepackage{multirow}
\usepackage{color}
\usepackage{cite}
\usepackage{diagbox}

\usepackage{amsthm}

\newtheorem{remark}{Remark}
\newtheorem{proposition}{Proposition}

\begin{document}
\title{Cyclic Delay-Doppler Shift: A Simple Transmit Diversity Technique for Ultra-Reliable Communications in Doubly Selective Channels}
\author{Haoran Yin, \textit{Graduate Student Member, IEEE}, Yu Zhou,  Yanqun Tang,  Di Zhang, \textit{Senior Member, IEEE}, Chi Zhang, Xizhang Wei, Jiaojiao Xiong,  Fan Liu, \textit{Senior Member, IEEE},   Marwa Chafii, \textit{Senior Member, IEEE},  and Mérouane Debbah, \textit{Fellow, IEEE}

\thanks{
		This work was supported in part by the Shenzhen Science and Technology Major Project under Grant KJZD20240903102000001 and in part by the Science and Technology Planning Project of Key Laboratory of Advanced IntelliSense Technology, Guangdong Science and Technology Department under Grant 2023B1212060024. An earlier version of this paper was presented in part at the IEEE International Conference on Acoustics, Speech and Signal Processing (ICASSP) Workshops 2023 \cite{bb23.3.6.1}. (\textit{Corresponding author: Yanqun Tang.})

		Haoran Yin, Yu Zhou, Yanqun Tang, and Jiaojiao Xiong are with the School of Electronics and Communication Engineering, Sun Yat-sen University, Shenzhen, China, and also with the Guangdong Provincial Key Laboratory of Sea-Air-Space Communication, Shenzhen, China (e-mail: \{yinhr6@mail2, zhouy633@mail2, tangyq8@mail, xiongjj7@mail2\}.sysu.edu.cn); 
		
		Chi Zhang and Xizhang Wei are with the School of Electronics and Communication Engineering, Sun Yat-sen University, Shenzhen, China (e-mail: \{zhangch397@mail2, weixzh7@mail\}.sysu.edu.cn).
		 
		 Di Zhang is with the School of Intelligent Systems Engineering, Sun Yat-sen University, Shenzhen, China (e-mail: zhangd263@mail.sysu.edu.cn).
		 
 		 Fan Liu is with the National Mobile Communications Research Laboratory, Southeast University, Nanjing, China (e-mail: fan.liu@seu.edu.cn).
		 
		 Marwa Chafii is with the Wireless Research Lab,
		 Engineering Division, New York University (NYU) Abu Dhabi, Abu Dhabi, UAE, and also with NYU WIRELESS, NYU Tandon School of Engineering, New York, USA (e-mail: marwa.chafii@nyu.edu).
		 
		 Mérouane Debbah is with the Electrical Engineering and Computer Science,
		 Khalifa University of Science and Technology, Abu Dhabi, UAE (e-mail: merouane.debbah@ku.ac.ae).		 
	}  
}

\maketitle

\begin{abstract}
	Affine frequency division multiplexing (AFDM) and orthogonal time frequency space (OTFS) are two promising advanced waveforms proposed for reliable communications in high-mobility scenarios.  In this paper, we introduce a simple transmit diversity technique, termed cyclic delay-Doppler shift (CDDS), for these two advanced waveforms to achieve ultra-reliable communications in doubly selective channels (DSCs). Two simple CDDS schemes, named modulation-domain CDDS (MD-CDDS) and time-domain CDDS (TD-CDDS), are proposed, which perform CDDS in advance at the transmitter before and after the modulation, respectively. We demonstrate that both of the two proposed CDDS schemes can be implemented efficiently and flexibly by multiplying the transmit vector with a well-designed precoding matrix, which is nothing but a sparse phase-compensated permutation matrix. Moreover, we theoretically and numerically prove that CDDS can provide MIMO-AFDM and MIMO-OTFS with optimal transmit diversity gain when a proper CDDS step is adopted. Compared to the conventional transmit diversity techniques, the proposed CDDS scheme enjoys the advantages of lower channel estimation overhead, implementation complexity, and signal processing latency, making it particularly suitable for ultra-reliable communications in high-mobility scenarios.
\end{abstract}


\section{Introduction}
The next-generation wireless networks (NGWNs) are conceived to support ultra-reliable, high-efficiency, and low-latency communication in high-mobility scenarios, including vehicle-to-vehicle (V2V), unmanned aerial vehicles (UAVs), and satellite networks \cite{bb25.6.10.1}.  The underlying multipath wireless channels in these advanced applications exhibit prominent Doppler effect due to their high-mobility nature, resulting in both time selectivity and frequency selectivity. In particular, Doppler shifts greatly impair the orthogonality among the subcarriers in conventional orthogonal frequency division multiplexing (OFDM) waveform at the receiver \cite{bb23.1.3.2}. Considering that higher frequency bands will be adopted in NGWNs and Doppler shifts scale linearly with the carrier frequency, designing a new waveform to better adapt to the doubly selective channels (DSCs) is of extreme importance. 
 
Against this background, one of the most promising alternatives is delay-Doppler (DD) waveforms, such as orthogonal time frequency space (OTFS) \cite{ bb23.2.7.1, bb24.08.21.2}, Zak-OTFS \cite{bb25.01.21.2}, and  orthogonal delay-Doppler division multiplexing (ODDM) \cite{bb24.03.15.2}. The core principle of OTFS is modulating information symbols in the DD domain by exploring the two-dimensional \emph{symplectic finite Fourier transform} (SFFT). Since the DD domain channel representation of the DSCs is much sparser and more stationary than the conventional time-frequency (TF) domain in OFDM, OTFS is resilient to the high DD shifts and is shown to significantly outperform OFDM in terms of bit error ratio (BER) \cite{bb23.2.12.1, 23.11.18.2, bb23.2.10.2}. 

Recently, another promising candidate named affine frequency division multiplexing (AFDM) was proposed, attracting substantial attention \cite{bb23.1.3.4,bb25.01.08.2,bb24.9.08.113,bb24.9.08.133, bb24.9.08.143, bb24.9.08.103, bb24.9.08.144}.  Information symbols in AFDM are multiplexed on a set of orthogonal chirps via the inverse \emph{discrete affine Fourier transform} (DAFT) \cite{ bb23.1.3.6}, which is characterized by two fundamental parameters that render AFDM compelling flexibility. By appropriately tuning the chirp slopes of the chirp subcarriers according to the Doppler profile of the DSC, AFDM manages to separate all propagation paths with distinct delay or Doppler shifts in the underlying one-dimensional DAFT domain, which guarantees optimal diversity gain in DSCs. In particular, the DAFT-domain channel representation in AFDM exhibits similar stability, sparsity, and separability as the DD-domain counterpart in OTFS, facilitating the channel estimation \cite{bb23.2.10.7,bb25.01.25.1, bb25.01.25.13} and signal detection \cite{bb25.01.25.3, bb24.9.23.2, bb24.9.23.201} of AFDM in DSCs. Extensive research has been conducted to explore the potential of AFDM in low-complexity modulation \cite{bb24.9.08.3, bb24.9.08.2},
  pulse shaping \cite{bb24.9.15.1}, multiple access \cite{bb24.9.08.1, bb24.9.08.10, bb25.01.08.2}, and integrated sensing and communications (ISAC) \cite{bb24.9.08.100, bb24.9.08.42, bb24.9.08.4, bb24.03.15.5, bb24.9.08.105, bb25.10.14.1, bb25.10.13.3}, and physical-layer security \cite{bb25.10.11.1,bb25.10.11.3,bb25.10.11.2,bb25.10.11.4}.

The diversity order of a waveform is a key indicator of its capacity to support reliable communications. It is mathematically defined as the negative slope of the curve of BER versus signal-to-noise ratio (SNR) on a logarithmic scale, characterizing the rate at which the BER decreases as SNR increases. One effective approach to increase the diversity order is to explore the space dimension by leveraging the multiple-input multiple-output (MIMO) techniques, which will continue to play a crucial role in NGWNs. Extensive works on MIMO-AFDM and MIMO-OTFS have been conducted, including channel estimation \cite{bb23.2.10.7, bb25.5.20.1}, signal detection \cite{bb25.10.04.3,bb25.10.13.1, bb25.10.04.2}, and ISAC \cite{bb24.9.08.104,bb25.10.04.1}. In particular, the authors in \cite{bb23.2.12.1} and \cite{bb23.2.10.7} demonstrated that optimal receive diversity gain equivalent to the number of receive antennas (RAs) can be straightforwardly obtained in MIMO-AFDM and MIMO-OTFS by simply applying joint-receive-antenna signal detection. However, compared to the receive diversity, the acquisition  of transmit diversity in MIMO-AFDM and MIMO-OTFS is of great challenge.

The classic Alamouti space-time coding (STC) transmit diversity technique initially proposed for single-carrier modulation in \cite{bb23.2.25.1} received considerable extensions in the last two decades. It encodes multiple symbols across multiple transmit antennas (TAs) over multiple time slots, leveraging time diversity and ensuring simple linear decoding in MIMO systems. Furthermore, the adaptation of Alamouti STC to MIMO-OTFS systems with two TAs was first investigated in \cite{bb23.2.10.8}, where the transmit diversity gain was shown to be two. However, it requires strict cooperation between two consecutive frames and assumes that the DSC remains unchanged, which is impractical in high-mobility scenarios. To address this limitation, the authors in \cite{bb23.2.10.9} divided the OTFS frame along the delay domain into two equivalent parts so that the Alamouti STC can be performed within a single transmit frame. Nevertheless, the extra guard symbols needed for frame division incur severe spectral efficiency degradation. Similar challenges arise when applying Alamouti STC to MIMO-AFDM, indicating that Alamouti STC may not be an appropriate choice for MIMO-AFDM and MIMO-OTFS to acquire transmit diversity gain in DSCs.

In addition to Alamouti STC, the cyclic delay diversity (CDD) \cite{bb25.5.26.3, bb23.2.26.1} and discontinuous Doppler diversity (DDoD) \cite{bb23.2.25.3, bb23.2.25.5} are two prominent transmit diversity schemes. Specifically, CDD applies fixed cyclic delay shifts to the transmitted signal at multiple TAs to obtain delay diversity, whereas
DDoD performs fixed Doppler shifts to obtain the Doppler diversity. Despite the attractive characteristics of low complexity and high effectiveness of CDD and DDoD, most of the works of CDD and DDoD mainly focused on MIMO-OFDM in frequency-selective channels, and their design principles and conclusions cannot be directly extended to advanced waveforms in DSCs. Insightfully, the authors in \cite{bb23.2.10.10} preliminarily combined CDD with MIMO-OTFS, achieving notable improvement in BER compared to single-input single-output (SISO) OTFS thanks to OTFS’s inherent capability to harvest delay-domain diversity. Although CDD and DoDD cannot guarantee optimal diversity gain equivalent to the number of TAs for MIMO-AFDM and MIMO-OTFS, their natural compatibility with these two waveforms shows great potential in enabling ultra-reliable communications in DSCs.

In this paper, motivated by the inborn capability of AFDM and OTFS to gather DD-domain diversity, we introduce a novel transmit diversity scheme, termed cyclic delay-Doppler shift (CDDS), for MIMO-AFDM and MIMO-OTFS to achieve ultra-reliable communications in DSCs. The core idea of CDDS is performing dedicated cyclic DD shifts in advance at different TAs according to the DD profile of the DSC to effectively augment the number of propagation paths of the wireless channel. Specifically, we propose two types of CDDS, which can provide optimal transmit diversity gain for MIMO-AFDM and MIMO-OTFS with appropriate CDDS configurations. We demonstrate with rigorous derivation that both CDDS schemes can be implemented by simply multiplying the transmit vector with a well-designed CDDS precoding matrix, which is nothing but a sparse, phase-compensated permutation matrix. Moreover, the multiple-TA system with CDDS is equivalent to a single-TA system, which significantly reduces the channel estimation overhead. Additionally, CDDS is performed within a single transmit frame and requires no additional processing at the receiver, which brings in the advantages of lower implementation complexity and signal processing latency. Our contributions are summarized as follows.
\begin{itemize} 
	\item[$\bullet$]
	We propose two types of CDDS, namely modulation-domain CDDS (MD-CDDS) and time-domain CDDS (TD-CDDS), which conduct CDDS before and after modulation, respectively. In particular, we derive the corresponding precoding matrices, noting that the TD-CDDS precoding matrix is suitable for all waveforms, whereas the MD-CDDS precoding matrix should be customized according to the input-output relationship (IOR) in the modulation domain of the specific waveform. Moreover, we unveil with a graphic interpretation that the CDDS operation can be considered as effectively shifting the channel delay and Doppler shifts, making the multiple TA system with CDDS equivalent to a single TA system with an increased number of propagation paths.
\end{itemize}
\begin{itemize} 
	\item[$\bullet$]
	We provide a comprehensive evaluation of the proposed CDDS schemes. Specifically, we prove that when all CDDS-shifted DD profiles remain non-overlapping, optimal transmit diversity gain is guaranteed for  MIMO-AFDM and MIMO-OTFS. Additionally, we make a full comparison between CDDS and conventional Alamouti schemes in terms of channel estimation overhead, showing that CDDS requires much fewer guard symbols, especially when the number of TAs is large. Based on that, we reveal that there is a fundamental tradeoff between the achieved transmit diversity and the channel estimation overhead, and CDDS can achieve a more compelling balance than conventional CDD and DDoD schemes thanks to its extra degree of freedom that enables more flexible CDDS step selection.
\end{itemize}
\begin{itemize} 
	\item[$\bullet$]
	Finally, extensive Monte Carlo simulations are conducted to investigate the performance of the two proposed CDDS schemes in terms of BER in DSCs with both integer and fractional Doppler shifts. It is shown that both CDDS schemes effectively provide MIMO-AFDM and MIMO-OTFS systems with optimal transmit diversity gain and advantage on lower channel estimation overhead over conventional transmit diversity schemes. Moreover, the performance of CDDS is robust to the adopted pulse shaping and imperfect channel state information (CSI).
\end{itemize}

The rest of this paper is organized as follows. Sec. \ref{Sec2} reviews the basic concepts of MIMO-AFDM and MIMO-OTFS systems, which lays the foundations for the demonstration of CDDS in Sec. \ref{Sec3}.  Sec. \ref{Sec4} provides a comprehensive investigation on the performance of CDDS, which is verified by the numerical results presented in Sec. \ref{Sec5}. Finally, Sec. \ref{Sec6} concludes this paper.

\textit{Notations:} Upper and lower case boldface letters denote matrices and column vectors, respectively; $\mathbb{C}$ denotes the set of complex numbers, and $\mathbb{C}^{M \times N}$ denotes the set of all $M \times N$ matrices with complex entries;  $\mathbf{I}_{N}$ denotes the identity matrix
of size $N \times N$;
$\mathbf{a} \sim \mathcal{C} \mathcal{N}\left(\mathbf{0}, N_{0} \mathbf{I}_{N}\right)$ means that $\mathbf{a}$ follows the complex Gaussian distribution with
zero mean and covariance $N_{0} \mathbf{I}_{N}$;  $\operatorname{diag}(\cdot)$ denotes a square diagonal matrix with the elements of the input vector on the main diagonal; 
$(\cdot)^{*}$,
$(\cdot)^{H}$, $(\cdot)^{T}$, and $\left\| \cdot \right\|$ denote the conjugate,
transpose, the Hermitian and the Euclidean norm  operations; $\lvert\cdot\rvert$ denotes the absolute value of a
complex scalar; $(\cdot)_{N}$ denotes the modulus operation with respect to $N$;
$\mathbb{E}[\cdot]$ denotes  the expectation; $Q(\cdot)$ denotes the tail distribution
function of the standard normal distribution;  $\text{vec}(\cdot)$ and $\text{vec}^{-1}(\cdot)$ denote the column-wise vectorization and its reverse operation, respectively; $\operatorname{card}(\cdot)$ denotes the cardinality.

\section{Fundamental of MIMO-AFDM and MIMO-OTFS}
\label{Sec2}
In this section, the basic concepts of MIMO-AFDM and MIMO-OTFS from \cite{bb23.2.7.1, bb23.1.3.4,bb25.01.08.2, bb24.08.21.2} are briefly reviewed. Fig.~\ref{2-1} shows the general modulation/demodulation block diagrams of MIMO-AFDM and MIMO-OTFS systems. Let $\mathbf{x}_{i}\ ( i=1,\dots,N_{t})$ and $\mathbf{y}_{j} \ ( j=1,\dots,N_{r})$ denote the modulation-domain transmitted vector of the $i$th TA and the modulation-domain received vector of the $j$th RA, respectively, where $N_{t}$ and $N_{r}$ represent the number of TAs and RAs, respectively. The modulation domain refers to the domain in which symbols (data, pilot, and guard symbols) are multiplexed, namely the DAFT domain for AFDM and the DD domain for OTFS. Correspondingly, $\mathbf{s}_{i}  \ (i=1,\dots,N_{t})$ and $\mathbf{d}_{j} \ ( j=1,\dots,N_{r})$ represent the time-domain transmitted vector of the $i$th transmit antenna and the time-domain received vector of the $j$th receive antenna, respectively.

 \begin{figure}[tbp]	
	\centering
	\includegraphics[width=0.48\textwidth,height=0.106\textwidth]{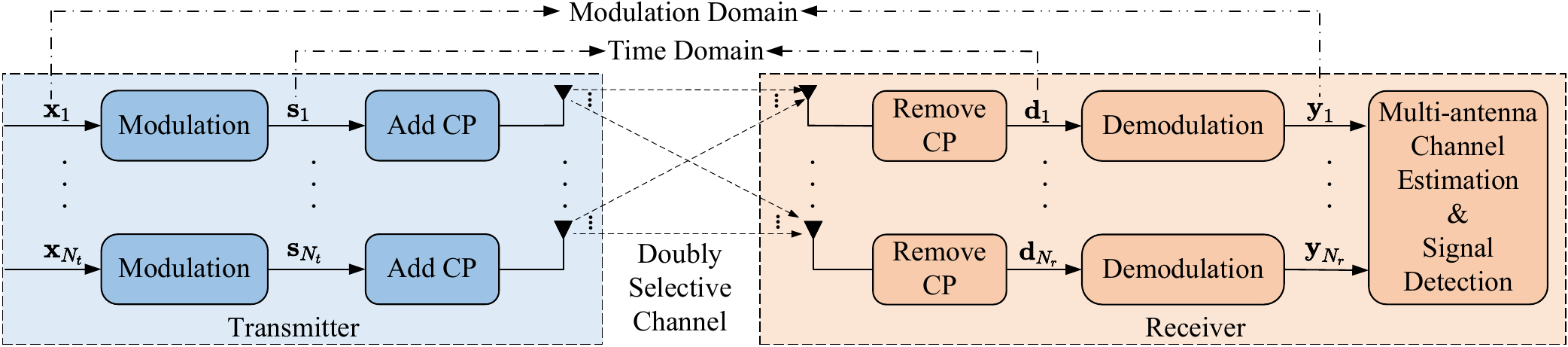}
	\caption{General modulation/demodulation block diagrams of baseband MIMO-AFDM and MIMO-OTFS systems.}
	\label{2-1}
\end{figure}

\subsection{AFDM Modulation and Demodulation}
Let $\mathbf{x}_{\text{AFDM}} \in \mathbb{A}^{N \times 1}$ denote a vector of $N$ quadrature amplitude modulation (QAM) symbols in the DAFT domain, where $N$ denotes the number of chirp subcarriers, $\mathbb{A}$ represents the modulation alphabet. At the transmitter, $N$-point inverse DAFT is firstly performed on $\mathbf{x}_{\text{AFDM}}$ to convert it into the time-domain signal $\mathbf{s}_{\text{AFDM}}$ as
\begin{equation}
	s_{\text{AFDM}}[n]=\frac{1}{\sqrt{N}} \sum_{m=0}^{N-1} x_{\text{AFDM}}[m] e^{j 2 \pi\left(c_{2} m^{2}+\frac{1}{N} m n+c_{1} n^{2}\right)},
	\label{eq2.12.5}
\end{equation}
where $n$ denotes the discrete-time domain, $m=0, \cdots, N-1$ denotes the index of AFDM chirp subcarrier,  and $c_{1}$ and $c_{2}$ are two fundamental parameters that determine the chirp slope and the initial phase of all AFDM chirp subcarriers, respectively. Equation (\ref{eq2.12.5}) can be written in matrix form as
\begin{equation}
	\mathbf{s}_{\text{AFDM}} =  \boldsymbol{\Lambda}_{c_{1}}^{H} \mathbf{F}^{H} \boldsymbol{\Lambda}_{c_{2}}^{H} \mathbf{x}_{\text{AFDM}} = \mathbf{A}^{H} \mathbf{x}_{\text{AFDM}}
	\label{eq2-2}
\end{equation}
where $\mathbf{A} = \boldsymbol{\Lambda}_{c_{2}} \mathbf{F} \boldsymbol{\Lambda}_{c_{1}}\in\mathbb{C}^{N\times N}$ represents the DAFT matrix,  $\mathbf{F}$ is the  discrete Fourier transform (DFT) matrix with entries  $e^{-j 2 \pi m n / N} / \sqrt{N}$,  $\boldsymbol{\Lambda}_{c}\triangleq\operatorname{diag}\left(e^{-j 2 \pi c q^{2}}, q=0,1, \ldots,   N-1\right)$ is a diagonal matrix. Before transmitting $\mathbf{s}_{\text{AFDM}}$ into the DSC, a chirp-periodic prefix (CPP) given by
\begin{align}
	s_{\text{AFDM}}[n] = s_{\text{AFDM}}[n+N]e^{-j 2 \pi c_{1}\left(N^{2}+2Nn \right)}, 
	\label{eq24.01.09.1}
\end{align}
$n=-L_{\text{CPP}}, \dots, -1$, is appended to $\mathbf{s}_{\text{AFDM}}$ to combat the multipath effect of the DSCs, where $L_{\text{CPP}}$ is the length of the CPP and should be set larger than or equal to the maximum delay spread of the channel. In particular, the CPP simplifies to a conventional cyclic prefix (CP) in OFDM when $2Nc_{1}$ is an integer and $N$ is even, as is considered throughout this paper.

After appending a CP, $\mathbf{s}_{\text{AFDM}}$ is transmitted into the DSC, which can be represented in the DD domain with delay $\tau$ and Doppler $\nu$ as
\begin{equation}
	h(\tau, \nu)=\sum_{i=1}^{P} h_{i} \delta\left(\tau-l_{i}\Delta t\right) \delta\left(\nu-k_{i}\Delta f\right),
	\label{eq2.12.2}
\end{equation}
where $P$ is the number of propagation paths, $h_{i}$ is the channel gain of the $i$th path,
$l_{i} \in [0, l_{\max}]$ and $k_{i} \in [-k_{\max}, k_{\max}]$ are assumed to be integers and represent normalized delay and Doppler shifts, respectively, with $l_{\max}$ and $k_{\max}$ denoting the maximum normalized delay and Doppler shifts, respectively. Moreover, $\Delta t$ and $\Delta f$ represent the Nyquist sampling interval and the chirp subcarrier spacing of AFDM, respectively, satisfying $\Delta f = \frac{1}{N\Delta t}$.

At the receiver, the CP component of the received time-domain signal is first removed, yielding the CP-free time-domain signal $\mathbf{d}_{\text{AFDM}}$ as
\begin{equation}
	d_{\text{AFDM}}[n]=\sum_{i=1}^{P} h_{i} e^{j  \frac{2 \pi}{N}k_{i}n} s_{\text{AFDM}}[(n-l_{i})_{N}]+\tilde{w}[n],
	\label{eq2-5}
\end{equation}
where  $\tilde{w} \sim \mathcal{C} \mathcal{N}\left(0, N_{0}\right)$ represents the additive white Gaussian noise (AWGN). Equation (\ref{eq2-5}) can be expressed in matrix form as 
\begin{equation}
	\mathbf{d}_{\text{AFDM}} = \sum_{i=1}^{P} h_{i} \mathbf{\tilde{H}}_{i}  \mathbf{s}_{\text{AFDM}}  + \mathbf{w}= \mathbf{\tilde{H}} \mathbf{s}_{\text{AFDM}} + \mathbf{\tilde{w}},  
	\label{eq2.12.3}
\end{equation}
where $\mathbf{w} \sim \mathcal{C} \mathcal{N}\left(\mathbf{0}, N_{0} \mathbf{I}_{N}\right)$ is the time-domain noise vector, $\mathbf{\tilde{H}}=\sum_{i=1}^{P}h_{i}  \mathbf{\tilde{H}}_{i}\in\mathbb{C}^{N\times N}$ denotes the time-domain channel matrix, $\mathbf{\tilde{H}}_{i}=  \boldsymbol{\Delta}_{N}^{k_{i}} \boldsymbol{\Pi}_{N}^{l_{i}}$ represents the time-domain subchannel matrix of the $i$th path, $\boldsymbol{\Pi}_{N}$ denotes the $N \times N$ forward cyclic-shift matrix given by
\begin{equation}
	\boldsymbol{\Pi}_{N}=\left[\begin{array}{cccc}
		0 & \cdots & 0 & 1 \\
		1 & \cdots & 0 & 0 \\
		\vdots & \ddots & \ddots & \vdots \\
		0 & \cdots & 1 & 0
	\end{array}\right]_{N \times N},
	\label{eq202505081}
\end{equation}
$\boldsymbol{\Pi}_{N}^{l_{i}}$ models the $l_{i}$ delay shift, while the digital frequency shift matrix $\boldsymbol{\Delta}_{N}^{k_{i}} \triangleq \operatorname{diag}\left(e^{j \frac{2 \pi}{N} k_{i} n}, n=0,1, \ldots, N-1\right)$ models the $k_{i}$ Doppler shift of the channels. Then 
$N$-point DAFT is implemented on $\mathbf{d}_{\text{AFDM}}$ to transform it to the DAFT-domain signal $\mathbf{y}_{\text{AFDM}}$ as
\begin{equation}
	y_{\text{AFDM}}[m]=\frac{1}{\sqrt{N}} \sum_{n=0}^{N-1} d_{\text{AFDM}}[n] e^{-j 2 \pi\left(c_{2} m^{2}+\frac{1}{N} m n+c_{1} n^{2}\right)},
	\label{eq2025.5.8.2}
\end{equation}
 $0 \leq m \leq N-1$, whose matrix form is given by
 \begin{equation}
 	\mathbf{y}_{\text{AFDM}} = \boldsymbol{\Lambda}_{c_{2}} \mathbf{F} \boldsymbol{\Lambda}_{c_{1}} \mathbf{d}_{\text{AFDM}} = \mathbf{A} \mathbf{d}_{\text{AFDM}}.
 	\label{eq2-8}
 \end{equation}
 Finally, the IOR of AFDM in the DAFT domain is given by \cite{bb23.1.3.4}
\begin{equation}
	\begin{aligned}
		y_{\text{AFDM}}[m]=&\sum_{i=1}^{P} h_{i} e^{j \frac{2 \pi}{N}\left(N c_{1} l_{i}^{2}-m^{\prime} l_{i}+N c_{2}\left(m'^{2}-m^{2}\right)\right)} x_{\text{AFDM}}[m^{\prime}] \\
		&\qquad\qquad\qquad+w[m], \ \ m^{\prime}=\left(m+\operatorname{ind}_{i}\right)_{N},
	\end{aligned}
	\label{eq2.12.6}
\end{equation}
with index indicator of the $i$th path $\operatorname{ind}_{i} \triangleq(2 N c_{1} l_{i}-k_{i})_{N}$, $w[m]$ being the DAFT-domain AWGN. Notably, Equation (\ref{eq2.12.6}) indicates that the DSC induces DATF-domain symbol spreading at the receiver. The matrix form of (\ref{eq2.12.6}) can be obtained by substituting (\ref{eq2-2}) and  (\ref{eq2.12.3}) into (\ref{eq2-8}) as
\begin{equation}
	\mathbf{y}_{\text{AFDM}} = \sum_{i=1}^{P} h_{i} \mathbf{H}_{i}^{\text{AFDM}} \mathbf{x}_{\text{AFDM}} + \mathbf{w} = \mathbf{H}_{\text{AFDM}}\mathbf{x}_{\text{AFDM}} + \mathbf{w},
	\label{eq2-12-2}
\end{equation}
where $\mathbf{H}_{i}^{\text{AFDM}} = \mathbf{A} \mathbf{\tilde{H}}_{i} \mathbf{A}^{H}$ denotes the DAFT-domain subchannel matrix of the $i$th path, $\mathbf{H}_{\text{AFDM}} = \sum_{i=1}^{P}h_{i} \mathbf{H}_{i}^{\text{AFDM}}$ is the DAFT-domain channel matrix, $\mathbf{w}\sim \mathcal{C} \mathcal{N} \left(\mathbf{0}, N_{0} \mathbf{I}_{N}\right)$ is the DAFT-domain noise vector.
\begin{remark} 
	\textup{It has been proven in \cite{bb23.1.3.4} (Theorem 1) that the diversity order of SISO-AFDM $\rho_{\text{SISO-AFDM}}=P$, i.e., SISO-AFDM can achieve optimal diversity in DSCs as long as
		\begin{equation} 
			c_{1} = \frac{2 (k_{\max }+k_{\text{space}})+1}{2N},
			\label{eq7-03-1}
		\end{equation}
		$c_{2}$ is set as either an arbitrary irrational number or a rational number sufficiently smaller than $\frac{1}{2N}$ (spacing factor $k_{\text{space}}$ is a non-negative integer and will be illustrated later in Sec. \ref{subsec4-2}), and 
		\begin{equation}
			N\geq \left(l_{\max }+1\right)\left(2 \left(k_{\max }+k_{\text{space}}\right)+1\right).
			\label{eq3-27}
		\end{equation} }
	\label{Remark1}
\end{remark} 

\subsection{OTFS Modulation and Demodulation}
Let $\mathbf{X}_{\text{OTFS}} \in \mathbb{A}^{K \times L}$ denote a matrix of $KL$ QAM symbols in the DD domain, where $L$ and $K$ denote the number of samples in the delay and Doppler dimensions, respectively. For ease of fair comparison between AFDM and OTFS, $KL$ is set to $N$ and the same Nyquist sampling interval is adopted, which guarantees that the two systems occupy the same TF resource. At the transmitter, $KL$ information symbols  $\mathbf{X}_{\text{OTFS}}$ are firstly mapped to TF domain  via inverse SFFT as
\begin{equation}
Q_{\text{OTFS}}[n, u]=\frac{1}{\sqrt{KL}} \sum_{k=0}^{K-1} \sum_{l=0}^{L-1} X_{\text{OTFS}}[k, l] e^{j 2 \pi\left(\frac{n k}{K}-\frac{u l}{L}\right)},
	\label{eq2.12.1}
\end{equation}
where $k=0,\dots,K-1$, $l=0,\dots,L-1$, $n=0,\dots,K-1$, and $u=0,\dots,L-1$ denote the indices of the Doppler, delay, time, and frequency domain, respectively. Then the TF symbols $\mathbf{Q}_{\text{OTFS}}$ are mapped to the time domain through Heisenberg transform as 
\begin{equation}
s_{\text{OTFS}}(t)=\sum_{n=0}^{K-1} \sum_{u=0}^{L-1} Q_{\text{OTFS}}[n, u] g_{\text{tx}}(t-n T) e^{j 2 \pi u \Delta u(t-n T)},
\label{eq25.05.09.1}
\end{equation}
where $T=L\Delta t$, $\Delta u=K \Delta f$, $g_{\text{tx}}(t)$ denotes the transmit pulse filter.

After passing through the same DSC described in (\ref{eq2.12.2}), we have the received time-domain signal given by
\begin{align}
	d_{\text{OTFS}}(t)=\sum_{i=1}^{P} h_{i}e^{j  2 \pi k_{i}\Delta ft}s(t-l_{i}\Delta t)+w(t),  
	\label{eq25.10.08.1}
\end{align}
which is converted to the TF domain with Wigner transform as
\begin{equation}
 G_{\text{OTFS}}[n, u]=\int d_{\text{OTFS}}(t) g_{\text{rx}}^{*}(t-n T) e^{-j 2 \pi u \Delta u(t-n T)} d t,
\label{eq25.05.09.2}
\end{equation}
where $g_{\text{rx}}(t)$ is the receive pulse filter. Then, $\mathbf{G}_{\text{OTFS}}$
is mapped back to DD-domain symbols with SFFT transform, which is given by
\begin{equation}
	Y_{\text{OTFS}}[k, l]=\frac{1}{\sqrt{KL}} \sum_{n=0}^{K-1} \sum_{u=0}^{L-1} G_{\text{OTFS}}[n, u] e^{-j 2 \pi\left(\frac{n k}{K}-\frac{u l}{L}\right)}.
	\label{eq25.05.09.3}
\end{equation}
Finally, the IOR of OTFS with bi-orthogonal transmit and receive pulses\footnote{Ideal bi-orthogonal transmit and receive pulses are used for ease of derivation. However, practical rectangular transmit and receive pulse filters are considered in the simulation part to demonstrate the robustness of the proposed schemes \cite{bb26.3.16.1, bb25.5.9.1,bb25.5.9.2}.} in the DD domain is given by \cite{bb23.2.7.1}
\begin{equation}
	\begin{aligned}
Y_{\text{OTFS}}[k, l]=\sum_{i=1}^{P} h_{i} e^{-j 2 \pi  \frac{k_{i}l_{i}}{KL} } X_{\text{OTFS}}[(k-k_{i})_{K},&(l-l_{i})_{L}] \\
& +w[k,l],
\end{aligned}
	\label{eq2.12.4}
\end{equation}
where $w[k,l]$ represents the DD-domain AWGM. Notably, Equation (\ref{eq2.12.4}) indicates that the DSC induces DD-domain symbol spreading at the receiver.

\begin{remark} 
	\textup{It has been proven in that \cite{bb23.2.12.1} OTFS attains nearly optimal diversity, i.e., $\rho_{\text{SISO-OTFS}}\approx P$, in DSCs in the finite SNR region when the frame size $KL$ is sufficiently large.}
\end{remark}

\subsection{MIMO-AFDM and MIMO-OTFS Systems}
The noise-free time-domain matrix form of the IOR between all TAs and the $r$th RA in MIMO-AFDM and MIMO-OTFS can be denoted as
\begin{align}
		\mathbf{d}_{r} &=\mathbf{\tilde{H}}_{r,1} \mathbf{s}_{1}+\mathbf{\tilde{H}}_{r,2} \mathbf{s}_{2}+\cdots+\mathbf{\tilde{H}}_{r,N_{t}} \mathbf{s}_{N_{t}},
	\label{eq25.05.10.1}
\end{align}
where
\begin{equation}
	\mathbf{\tilde{H}}_{r,t}=\sum_{i=1}^{P} h_{i}^{[r,t]} \mathbf{\tilde{H}}_{i}^{[r,t]} 
	=\sum_{i=1}^{P} h_{i}^{[r,t]} \boldsymbol{\Delta}_{N}^{k_{i}} \boldsymbol{\Pi}_{N}^{l_{i}}
	\label{eq25.05.10.2}
\end{equation}
is the time-domain channel matrix between the $r$th RA and the $t$th TA, $\mathbf{\tilde{H}}_{i}^{[r,t]}$
and $h^{[r,t]}_{i}$ are the associated time-domain subchannel matrix and the channel gain of the $i$th path, respectively. Moreover, the noise-free modulation-domain matrix form IOR can be denoted as
\begin{align}
		\mathbf{y}_{r} &=\mathbf{H}_{r,1} \mathbf{x}_{1}+\mathbf{H}_{r,2} \mathbf{x}_{2}+\cdots+\mathbf{H}_{r,N_{t}} \mathbf{x}_{N_{t}},
	\label{eq100}
\end{align}
where $\mathbf{H}_{r,t}=\sum_{i=1}^{P}h_{i}^{[r,t]} \mathbf{H}_{i}^{[r,t]}$
represents the modulation-domain channel matrix (e.g., $\mathbf{H}_{\text{AFDM}}$ in (\ref{eq2-12-2}) for AFDM) between the $r$th RA and the $t$th TA, and $\mathbf{H}_{i}^{[r,t]}$ is the associated modulation-domain subchannel matrix of the $i$th path.

When each TA transmits independent information symbols, i.e., $\mathbf{x}_{1}\neq \mathbf{x}_{2}\neq \dots \neq\mathbf{x}_{N_{t}}$, MIMO systems explore the transmit multiplexing gain for linear increment in spectral efficiency with the number of TAs and the optimal receive diversity gain with the number of RAs. Instructively, the multiple TAs can also be leveraged to obtain transmit diversity gain by performing dedicated precoding on fewer symbols at multiple TAs, as is applied in the following proposed CDDS scheme, trading multiplexing gain (spectral efficiency) for diversity gain (reliability).

\section{Cyclic Delay-Doppler Shift}
\label{Sec3}
In this section, we propose TD-CDDS and MD-CDDS schemes for MIMO-AFDM and MIMO-OTFS, leveraging their optimal diversity properties to acquire transmit diversity gain. A detailed performance analysis will be provided in the next section. 

 \begin{figure}[tbp]
	\centering
	\includegraphics[width=0.45\textwidth,height=0.25\textwidth]{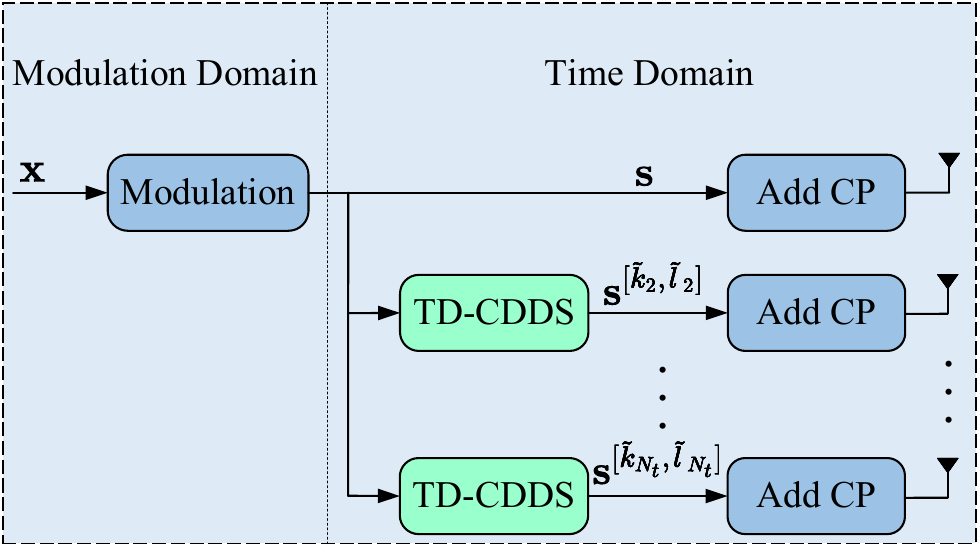}
	\caption{Block diagrams of the transmitter in a TD-CDDS system.}
	\label{3-1}
\end{figure}

\subsection{TD-CDDS}
We first introduce the TD-CDDS scheme, whose block diagram is shown in Fig. \ref{3-1}. TD-CDDS is implemented after modulation at the transmitter with $N_{t}$ TAs.  Let $\mathbf{x}$ represent the modulation-domain information vector to be transmitted and $\mathbf{s}$ represent the CP-free transmitted time-domain vector after modulation. At the $t$th TA ($t=2,\cdots,N_{t}$), we perform $\tilde{l}_{t}$-step cyclic delay shift and $\tilde{k}_{t}$-step cyclic Doppler shift on $\mathbf{s}$, referred to as $[\tilde{k}_{t},\tilde{l}_{t}]$-step TD-CDDS, by multiplying $\mathbf{s}$ with $\boldsymbol{\Pi}_{N}^{\tilde{l}_{t}}$ and $\boldsymbol{\Delta}_{N}^{\tilde{k}_{t}}$ successively as
\begin{equation}
	\mathbf{s}^{[\tilde{k}_{t},\tilde{l}_{t}]} = \boldsymbol{\Delta}_{N}^{\tilde{k}_{t}} \boldsymbol{\Pi}_{N}^{\tilde{l}_{t}}\mathbf{s}= \mathbf{C}_{\text{TD-CDDS}}^{[\tilde{k}_{t},\tilde{l}_{t}]}\mathbf{s},
	\label{eq2.14.1}
\end{equation}
with $\mathbf{C}_{\text{TD-CDDS}}^{[\tilde{k}_{t},\tilde{l}_{t}]}=\boldsymbol{\Delta}_{N}^{\tilde{k}} \boldsymbol{\Pi}_{N}^{\tilde{l}}$ being the $[\tilde{k}_{t},\tilde{l}_{t}]$-step TD-CDDS precoding matrix. Then  after adding a CP, $\mathbf{s}^{[\tilde{k}_{t},\tilde{l}_{t}]}$ is transmitted into the DSC at the $t$th TA.

Based on that, by substituting $\mathbf{s}, \mathbf{s}^{[\tilde{k}_{2},\tilde{l}_{2}]},\dots, \mathbf{s}^{[\tilde{k}_{N_{t}},\tilde{l}_{N_{t}}]}$ into $\mathbf{s}_{1}, \mathbf{s}_{2},\dots, \mathbf{s}_{N_{t}}$ in (\ref{eq25.05.10.1}), respectively, we have the received time-domain signal at the $r$th RA as
\begin{equation}
	\mathbf{d}_{r} =\underbrace{\mathbf{\tilde{H}}_{r,1} \mathbf{s}}_{\mathbf{d}_{r,1}}+\underbrace{\mathbf{\tilde{H}}_{r,2} \mathbf{s}^{[\tilde{k}_{2},\tilde{l}_{2}]}}_{\mathbf{d}_{r,2}}+\cdots+\underbrace{\mathbf{\tilde{H}}_{r,N_{t}} \mathbf{s}^{[\tilde{k}_{N_{t}},\tilde{l}_{N_{t}}]}}_{\mathbf{d}_{r,N_{t}}},
	\label{eq25.05.10.4}
\end{equation}
where $\mathbf{d}_{r,t}$ ($t=1, \dots, N_{t}$) is the component come from the $t$th TA and can be rewritten according to (\ref{eq25.05.10.2}) and (\ref{eq2.14.1}) as 
\begin{align}
	\mathbf{d}_{r,t} &= \mathbf{\tilde{H}}_{r,t}
	\mathbf{C}_{\text{TD-CDDS}}^{[\tilde{k}_{t},\tilde{l}_{t}]}
	\mathbf{s} \notag\\
	&=\sum_{i=1}^{P} h_{i}^{[r,t]} \boldsymbol{\Delta}_{N}^{k_{i}} \boldsymbol{\Pi}_{N}^{l_{i}}
	\boldsymbol{\Delta}_{N}^{\tilde{k}_{t}} \boldsymbol{\Pi}_{N}^{\tilde{l}_{t}}\mathbf{s}
	\label{eq25.05.12.1} \\
	&=\sum_{i=1}^{P} h_{i}^{[r,t]}
	e^{-j \frac{2 \pi}{N} \tilde{k}_{t} l_{i}}
	\boldsymbol{\Delta}_{N}^{k_{i}}
	\boldsymbol{\Delta}_{N}^{\tilde{k}_{t}} \boldsymbol{\Pi}_{N}^{l_{i}}
	 \boldsymbol{\Pi}_{N}^{\tilde{l}_{t}}\mathbf{s}
	 \label{eq25.05.12.2}\\
&=\sum_{i=1}^{P} \bar{h}_{i}^{[r,t]} \boldsymbol{\Delta}_{N}^{\bar{k}_{i}^{[t]}} \boldsymbol{\Pi}_{N}^{\bar{l}_{i}^{[t]}} 
\mathbf{s} \label{eq25.05.12.3}\\
&=\sum_{i=1}^{P}\bar{h}_{i}^{[r,t]} \mathbf{\tilde{\bar{{H}}}}_{i}^{[r,t]}\mathbf{s} =\mathbf{\tilde{\bar{H}}}_{r,t}\mathbf{s},
\label{eq24.01.09.2}
\end{align}
where the derivation from (\ref{eq25.05.12.1}) to (\ref{eq25.05.12.2}) is based on the fact that $\boldsymbol{\Pi}_{N}^{l_{i}}
\boldsymbol{\Delta}_{N}^{\tilde{k}_{t}}=	e^{-j \frac{2 \pi}{N} \tilde{k}_{t} l_{i}}\boldsymbol{\Delta}_{N}^{\tilde{k}_{t}} \boldsymbol{\Pi}_{N}^{l_{i}}$, $\bar{h}_{i}^{[r,t]}=h_{i}^{[r,t]}
e^{-j \frac{2 \pi}{N} \tilde{k}_{t} l_{i}}$, $ \bar{k}_{i}^{[t]}=k_{i}+\tilde{k}_{t}$, and $\bar{l}_{i}^{[t]}=l_{i}+\tilde{l}_{t}$ represent the effective channel gain, Doppler shift and delay shift of the $i$th path between the $r$th RA and the $t$th TA after $[\tilde{k}_{t},\tilde{l}_{t}]$-step TD-CDDS, respectively, $\mathbf{\tilde{\bar{{H}}}}_{i}^{[r,t]}$ and $\mathbf{\tilde{\bar{H}}}_{r,t}=\sum_{i=1}^{P}\bar{h}_{i}^{[r,t]} \mathbf{\tilde{\bar{{H}}}}_{i}^{[r,t]}$
are the associated effective time-domain subchannel matrix and channel matrix, respectively.  Subsequently, substituting (\ref{eq24.01.09.2}) into 
(\ref{eq25.05.10.4}), we have
\begin{align}
	\mathbf{d}_{r} &=\mathbf{\tilde{\bar{H}}}_{r,1}\mathbf{s}
	+\mathbf{\tilde{\bar{H}}}_{r,2}\mathbf{s}+\cdots+\mathbf{\tilde{\bar{H}}}_{r,N_{t}}\mathbf{s} \notag \\
	&=\sum_{t=1}^{N_{t}}\mathbf{\tilde{\bar{H}}}_{r,t} \mathbf{s}=\mathbf{\tilde{\bar{H}}}_{r}\mathbf{s}
	\label{eq25.05.12.4}
\end{align}
where $\mathbf{\tilde{\bar{H}}}_{r}=\sum_{t=1}^{N_{t}}\mathbf{\tilde{\bar{H}}}_{r,t}$.

\begin{remark} 
	\textup{Equation (\ref{eq25.05.12.3}) shows that the TD-CDDS operation on $\mathbf{s}$ in (\ref{eq2.14.1}) is equivalent to shifting the delay and Doppler shifts of all the original propagation paths simultaneously by $\tilde{l}_{t}$ and $\tilde{k}_{t}$, respectively\footnote{This is also feasible for AFDM with a CPP that cannot be simplified to a CP by simply performing phase compensation according to the extra phase in the CPP provided in (\ref{eq24.01.09.1}) and the adopted CDDS step.}. Since (\ref{eq2.14.1}) is performed in the time domain, we refer to it as TD-CDDS, which can be considered as the generalization of CDD and DDoD from the frequency-selective channels to the DSCs. Moreover, we can observe from (\ref{eq25.05.12.4}) that the MIMO system developed in (\ref{eq25.05.10.1}) essentially turns into an equivalent single-input multiple-output (SIMO) system after TD-CDDS, where $\mathbf{\tilde{\bar{H}}}_{r}$ serves as the equivalent time-domain channel matrix between the equivalent single TA and the $r$th RA.}
\end{remark}

\subsection{MD-CDDS}
\label{subsec4-2}
\begin{figure}[tbp]
	\centering
	\includegraphics[width=0.45\textwidth,height=0.25\textwidth]{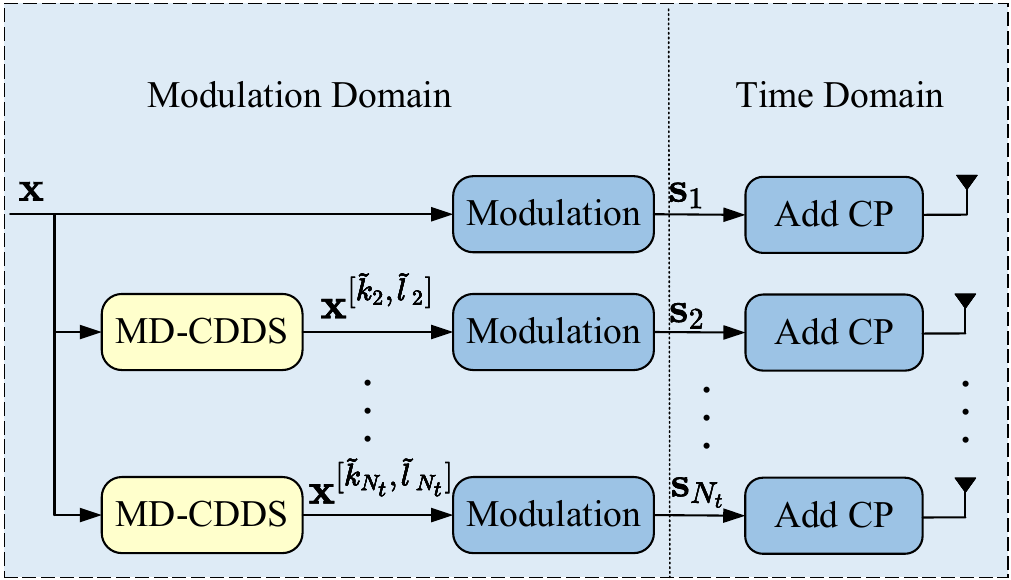}
	\caption{Block diagrams of the transmitter in an MD-CDDS system.}
	\label{3-2}
\end{figure}

We next illustrate the implementation of MD-CDDS in MIMO-OTFS and MIMO-AFDM systems. Different from TD-CDDS, MD-CDDS is performed before the modulation, as shown in Fig. \ref{3-2}. 

\subsubsection{\textbf{MD-CDDS-OTFS}} Let $\mathbf{x}_{\text{OTFS}}=\text{vec}(\mathbf{X}_{\text{OTFS}})$, then the $\tilde{l}_{t}$-step cyclic delay shift of $\mathbf{X}_{\text{OTFS}}$ can be obtained with $\text{vec}^{-1}(\boldsymbol{\Pi}_{KL}^{K*\tilde{l}_{t}}\mathbf{x}_{\text{OTFS}})$, and the $\tilde{k}_{t}$-step cyclic Doppler shift counterpart can be acquired with $\text{vec}^{-1}((\mathbf{I}_{L}
\otimes \boldsymbol{\Pi}_{K}^{\tilde{k}_{t}}) \mathbf{x}_{\text{OTFS}})$, where $\otimes$ denotes the Kronecker product operation.
Therefore, performing $[\tilde{k}_{t},\tilde{l}_{t}]$-step MD-CDDS on $\mathbf{X}_{\text{OTFS}}$ at the $t$th TA is equivalent to multiplying $\mathbf{X}_{\text{OTFS}}$ with $\boldsymbol{\Pi}_{KL}^{K*\tilde{l}_{t}}$ and $(\mathbf{I}_{L}
\otimes \boldsymbol{\Pi}_{K}^{\tilde{k}})$ successively as
\begin{align}
		\mathbf{X}_{\text{OTFS}}^{[\tilde{k}_{t},\tilde{l}_{t}]}
		&=\text{vec}^{-1}((\mathbf{I}_{L}
		\otimes \boldsymbol{\Pi}_{K}^{\tilde{k}})\boldsymbol{\Pi}_{KL}^{K*\tilde{l}_{t}}\mathbf{x}_{\text{OTFS}})
		 \notag \\
		&= \text{vec}^{-1}(\mathbf{C}_{\text{MD-CDDS}}^{[\tilde{k}_{t},\tilde{l}_{t}],\text{OTFS}}\mathbf{x}_{\text{OTFS}}),
	\label{eq2.14.3}
\end{align}
where  
\begin{align}
	\mathbf{C}_{\text{MD-CDDS}}^{[\tilde{k}_{t},\tilde{l}_{t}],\text{OTFS}}=(\mathbf{I}_{L}
	\otimes \boldsymbol{\Pi}_{K}^{\tilde{k}})\boldsymbol{\Pi}_{KL}^{K*\tilde{l}_{t}}
		\label{eq25.05.12.6}
\end{align} 
is the $[\tilde{k}_{t},\tilde{l}_{t}]$-step MD-CDDS precoding matrix of OTFS. Notably, we have
\begin{align}
	X_{\text{OTFS}}^{[\tilde{k}_{t},\tilde{l}_{t}]}[k,l] = X_{\text{OTFS}}[(k-\tilde{k}_{t})_{K},(l-\tilde{l}_{t})_{L}].
	\label{eq25.05.12.5}
\end{align}
Consequently, the noise-free IOR of OTFS in (\ref{eq2.12.4}) becomes
\begin{align}	
		&Y_{\text{OTFS}}^{[r,t]}[k, l] = 
		\sum_{i=1}^{P} h_{i}^{[r,t]} e^{-j 2 \pi \frac{k_{i}l_{i}}{KL}}		X_{\text{OTFS}}^{[\tilde{k}_{t},\tilde{l}_{t}]}[(k-k_{i})_{K},(l-l_{i})_{L}] \notag \\
		&\overset{(\ref{eq25.05.12.5})}{=}\sum_{i=1}^{P} h_{i}^{[r,t]}
		e^{j 2 \pi \frac{k_{i}\tilde{l}_{t}+\tilde{k}_{t}l_{i}+\tilde{k}_{t}\tilde{l}_{t}}{KL}} 
		e^{-j 2 \pi \frac{(k_{i} + \tilde{k}_{t} )(l_{i}+ \tilde{l}_{t})}{KL}} \notag \\
		 & \qquad  \qquad \qquad \times X_{\text{OTFS}}[(k-(k_{i}+\tilde{k}_{t}))_{K},(l-(l_{i}+\tilde{l}_{t}))_{L}] \notag \\
		 & = \sum_{i=1}^{P} \bar{h}_{i, \text{OTFS}}^{[r,t]}
		 e^{-j 2 \pi \frac{\bar{k}_{i}^{[t]}\bar{l}_{i}^{[t]}}{KL}}
		 X_{\text{OTFS}}[(k-\bar{k}_{i}^{[t]})_{K},(l-\bar{l}_{i}^{[t]})_{L}],
	\label{eq2.14.4}
\end{align}
where $\mathbf{Y}_{\text{OTFS}}^{[r,t]}$ is the received DD-domain symbol matrix of the $r$th RA that contributed by the $t$th TA, $\bar{h}_{i, \text{OTFS}}^{[r,t]} = h_{i}^{[r,t]}
e^{j 2 \pi \frac{k_{i}\tilde{l}_{t}+\tilde{k}_{t}l_{i}+\tilde{k}_{t}\tilde{l}_{t}}{KL}}$, $ \bar{k}_{i}^{[t]}=k_{i}+\tilde{k}_{t}$, and $\bar{l}_{i}^{[t]}=l_{i}+\tilde{l}_{t}$ represent the effective channel gain, Doppler shift and delay shift of the $i$th path between the $r$th RA and the $t$th TA after $[\tilde{k}_{t},\tilde{l}_{t}]$-step MD-CDDS, respectively.

\subsubsection{\textbf{MD-CDDS-AFDM}} For AFDM with the parameter setting of (\ref{eq7-03-1}), a full DD representation can be obtained, i.e., all the paths can be sufficiently separated in the DAFT domain \cite{bb23.1.3.4}. Fig. \ref{3-3} demonstrates the bijective relationship between the DD-domain and DAFT-domain channel representations in OTFS and AFDM, where the latter can be considered as splicing multiple $[2(k_{\max}+ k_{\text{space}})+1]$-length delay blocks with each delay block is separated from its adjacence by two $k_{\text{space}}$-length spacing band. Insightfully, this bijective relationship indicates that a unit Doppler shift results in a unit cyclic shift in the DAFT domain, whereas a unit delay shift corresponds to a $(-[2(k_{\max}+ k_{\text{space}})+1])$-step cyclic shift in the DAFT domain \cite{bb25.01.08.2}. This property aligns with the DAFT-domain symbol spreading demonstrated in the AFDM IOR in (\ref{eq2.12.6}). Consequently, $[\tilde{k}_{t},\tilde{l}_{t}]$-step MD-CDDS in AFDM corresponds to performing $\tilde{m}_{t}$-step shifts in the DAFT domain, where 
\begin{align}
\tilde{m}_{t} &= \tilde{k}_{t}-[ 2(k_{\max}+ k_{\text{space}})+1] \tilde{l}_{t}\notag \\
&= \tilde{k}_{t}-2Nc_{1}\tilde{l}_{t}.
	\label{eq25.05.12.10}
\end{align} 
Besides, to ensure that the indices of the new paths after cyclic Doppler shifts remain in the original delay blocks, the parameter $k_{\text{space}}$ in $c_{1}$ should satisfy
\begin{align}
	k_{\text{space}} \geq  \max\{k_{\max}, k_{2,\max},\dots,k_{N_{t},\max}\}-k_{\max},
	\label{eq25.05.12.8}
\end{align} 
where
\begin{align}
k_{t,\max}\triangleq\max\{|k_{1}+\tilde{k}_{t}|,\dots,|k_{P}+\tilde{k}_{t}|\}, \ t=2, \dots, N_{t}
	\label{eq25.05.12.9}
\end{align} 
 denotes the effective maximum Doppler after $\tilde{k}_{t}$-step cyclic Doppler shift.
 
\begin{figure}[tbp]
 	\centering
 	\includegraphics[width=0.47\textwidth,height=0.17\textwidth]{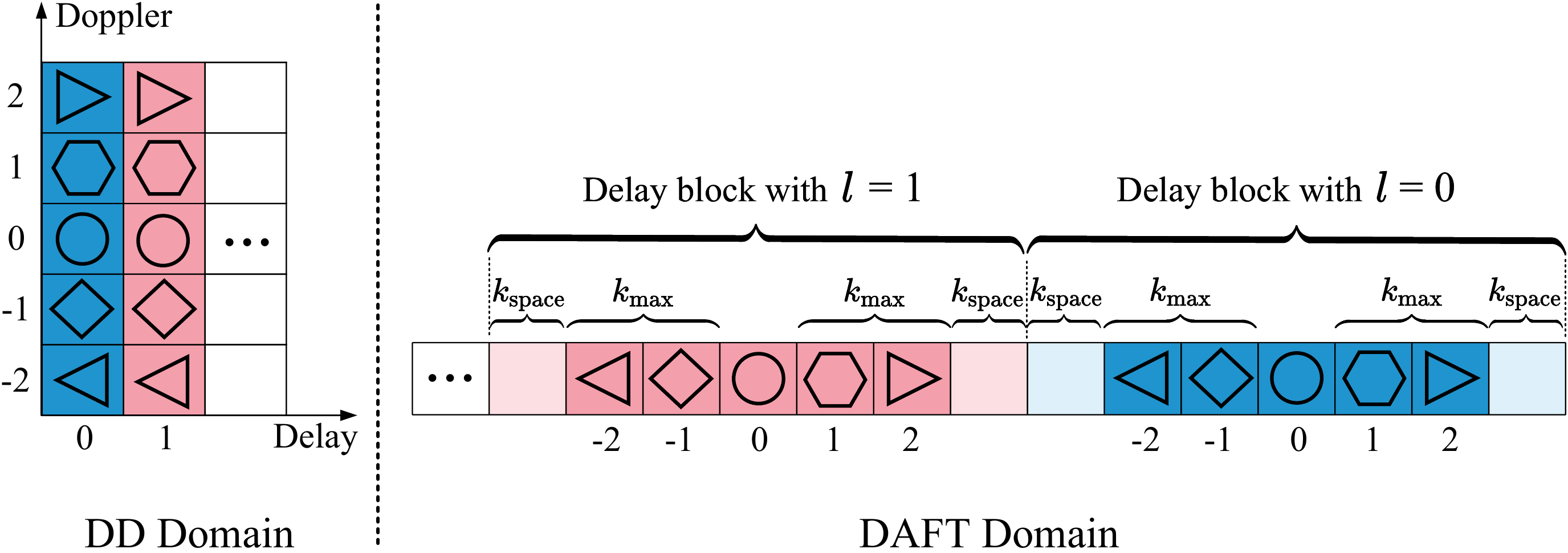}
 	\caption{The bijective relation between DD-domain (left) and DAFT-domain (right)  channel representations ($k_{\max}=2$ and $k_{\text{space}}=1$).}
 	\label{3-3}
 \end{figure}

Building on these insights, we proceed to derive the MD-CDDS matrix of AFDM. To this end, we first multiple $\mathbf{x}_{\text{AFDM}}$ with $\boldsymbol{\Pi}_{N}^{\tilde{m}_{t}}$ to perform $\tilde{m}_{t}$-step cyclic DAFT-domain shift as 
\begin{align}
	\mathbf{x}_{\text{AFDM}}^{[\tilde{m}_{t}]} = 	\boldsymbol{\Pi}_{N}^{\tilde{m}_{t}}\mathbf{x}_{\text{AFDM}},
	\label{eq25.05.12.12}
\end{align} 
which implies that
\begin{align}
	\mathbf{x}_{\text{AFDM}}^{[\tilde{m}_{t}]}[m] = 	\mathbf{x}_{\text{AFDM}}[(m-\tilde{m}_{t})_{N}].
	\label{eq25.05.12.13}
\end{align} 
Then the noise-free IOR in (\ref{eq2.12.6}) converts to (\ref{eq2.16.1}) (shown at the top of the next page), where  $\mathbf{y}_	{\text{AFDM}}^{[r,t]}$ is the received DAFT-domain symbol of the $r$th RA that contributed by the $t$th TA. Substituting (\ref{eq25.05.12.13}) into (\ref{eq2.16.1}), we have (\ref{eq25.05.14.1}). Let $\bar{m} = (m'-\tilde{m}_{t})_{N}$, we obtain $m' = (\bar{m}+ \tilde{m}_{t})_{N}$, which yields (\ref{eq25.05.14.2}), where
\begin{align}
	\bar{m}& = (\left(m+\operatorname{ind}_{i}\right)_{N}-\tilde{m}_{t})_{N}\notag\\
	&=(m+2 N c_{1} l_{i}-k_{i}-(\tilde{k}_{t}-2Nc_{1}\tilde{l}_{t}))_{N}\notag\\
	&=(m+2 N c_{1} (l_{i}+\tilde{l}_{t})-(k_{i}+\tilde{k}_{t}))_{N}\notag\\
	&=(m+\bar{\operatorname{ind}}_{i}^{[t]})_{N}
	\label{eq25.05.12.15}
\end{align} 	
and $
	\bar{\operatorname{ind}}_{i}^{[t]} \triangleq(2 N c_{1}\bar{l}_{i}^{[t]}-\bar{k}_{i}^{[t]})_{N}$ is the effective index indicator of the $i$th path. After some algebraic manipulations, we arrive at (\ref{eq25.05.15.2}), where
\begin{figure*}
\begin{align}
			y_{\text{AFDM}}^{[r,t]}[m]&=\sum_{i=1}^{P}  h_{i}^{[r,t]} e^{j \frac{2 \pi}{N}\left(N c_{1} l_{i}^{2}-m^{\prime} l_{i}+N c_{2}\left(m'^{2}-m^{2}\right)\right)} x_{\text{AFDM}}^{[\tilde{m}_{t}]}[m^{\prime}], \ \ m^{\prime}=\left(m+\operatorname{ind}_{i}\right)_{N} \label{eq2.16.1} \\
			&\overset{(\ref{eq25.05.12.13})}{=}\sum_{i=1}^{P}  h_{i}^{[r,t]} e^{j \frac{2 \pi}{N}\left(N c_{1} l_{i}^{2}-m^{\prime} l_{i}+N c_{2}\left(m'^{2}-m^{2}\right)\right)} x_{\text{AFDM}}[(m'-\tilde{m}_{t})_{N}], \ \ m^{\prime}=\left(m+\operatorname{ind}_{i}\right)_{N} \label{eq25.05.14.1} \\
			&=\sum_{i=1}^{P}  h_{i}^{[r,t]} e^{j \frac{2 \pi}{N}\left(N c_{1} l_{i}^{2}-(\bar{m}+ \tilde{m}_{t})_{N} l_{i}+N c_{2}\left((\bar{m}+ \tilde{m}_{t})_{N}^{2}-m^{2}\right)\right)} x_{\text{AFDM}}[\bar{m}], \ \ \bar{m}=(m+\bar{\operatorname{ind}}_{i}^{[t]})_{N} \label{eq25.05.14.2} \\
			&=\sum_{i=1}^{P} 
			 h_{i}^{[r,t]} 
			e^{j \frac{2 \pi}{N}\left(-Nc_{1}(2l_{i}\tilde{l}_{t}+\tilde{l}_{t}^{2})+ \left(\bar{m}-(\bar{m}+\tilde{m}_{t})_{N}\right)l_{i}+\bar{m}\tilde{l}_{t}+Nc_{2}\left(\left(\bar{m}+\tilde{m}_{t}\right)_{N}^{2}-\bar{m}^{2}\right)\right)} \notag\\
			& \qquad  \quad \qquad \qquad  \qquad \qquad \qquad \times e^{j \frac{2 \pi}{N}\left(N c_{1} (l_{i}+\tilde{l}_{t})^{2}-\bar{m}(l_{i}+\tilde{l}_{t})+N c_{2}\left(\bar{m}^{2}-m^{2}\right)\right)} x_{\text{AFDM}}[\bar{m}], \ \ \bar{m}=(m+\bar{\operatorname{ind}}_{i}^{[t]})_{N} \notag\\
			&=\sum_{i=1}^{P} 
			 h_{i}^{[r,t]}
			\underbrace{e^{-j \frac{2 \pi}{N}\left(Nc_{1}(2l_{i}\tilde{l}_{t}+\tilde{l}_{t}^{2})
				+\tilde{m}_{t}l_{i}\right)}}_{\mathcal{A}(l_{i}, \tilde{m}_{t})}
			\underbrace{e^{j \frac{2 \pi}{N}\left( 
				\bar{m}\tilde{l}_{t}+Nc_{2}\left(\left(\bar{m}+\tilde{m}_{t}\right)_{N}^{2}-\bar{m}^{2}\right)
				\right)}}_{\mathcal{E}(\bar{m}, \tilde{m}_{t})}
			\notag\\
			& \qquad  \qquad \quad  \qquad \qquad  \qquad \qquad \qquad \times e^{j \frac{2 \pi}{N}\left(N c_{1} (\bar{l}_{i}^{[t]})^{2}-\bar{m}\bar{l}_{i}^{[t]}+N c_{2}\left(\bar{m}^{2}-m^{2}\right)\right)} x_{\text{AFDM}}[\bar{m}], \ \ \bar{m}=(m+\bar{\operatorname{ind}}_{i}^{[t]})_{N}. \label{eq25.05.15.2} 
\end{align} 
\hrulefill
\end{figure*}
\begin{align}
	\mathcal{A}(l_{i}, \tilde{m}_{t}) = 
	e^{-j \frac{2 \pi}{N}\left(Nc_{1}(2l_{i}\tilde{l}_{t}+\tilde{l}_{t}^{2})
		+\tilde{m}_{t}l_{i}\right)}
\end{align}
and 
\begin{equation}
	\mathcal{E}(\bar{m}, \tilde{m}_{t}) = e^{j \frac{2 \pi}{N}\left( 
		\bar{m}\tilde{l}_{t}+Nc_{2}\left(\left(\bar{m}+\tilde{m}_{t}\right)_{N}^{2}-\bar{m}^{2}\right)
		\right)}.
	\label{eq2.16.3}
\end{equation}
It is important to notice that $\mathcal{E}(m^{\prime}, \tilde{m}_{t})$ has no relevance to the channel parameters $h_{i}, l_{i}$, and $k_{i}$ and received DAFT-domain index $m$, which means we can compensate it in advance at the transmitting end by multiplying $\mathbf{x}_{\text{AFDM}}$ with a phase compensation  matrix 
\begin{align}
\mathbf{P}_{N}^{[\tilde{m}_{t}]} = \operatorname{diag}\left(\mathcal{E}^{*}(\bar{m}, \tilde{m}_{t}), \bar{m}=0,1, \ldots, N-1\right)
\end{align}
before performing $\tilde{m}_{t}$-step cyclic DAFT-domain shift. 
\begin{figure*}
	\begin{equation}
		y_{\text{AFDM}}^{[r,t]}[m]=\sum_{i=1}^{P} \bar{h}_{i,\text{AFDM}}^{[r,t]} e^{j \frac{2 \pi}{N}\left(N c_{1} (\bar{l}_{i}^{[t]})^{2}-m^{\prime} \bar{l}_{i}^{[t]}+N c_{2}\left(m'^{2}-m^{2}\right)\right)} x_{\text{AFDM}}[\bar{m}], \ \ \bar{m}=\left(m+\bar{\operatorname{ind}}_{i}^{[t]} \right)_{N}, \ 0 \leq m \leq N-1.
		\label{eq2.16.2}
	\end{equation}
	\hrulefill
\end{figure*}
By doing so, (\ref{eq25.05.15.2}) can be further simplified as (\ref{eq2.16.2}), where $\bar{h}_{i,\text{AFDM}}^{[r,t]} = h_{i}^{[r,t]}\mathcal{A}(l_{i}, \tilde{m}_{t})$ is the effective channel gain of the $i$th path between the $r$th RA and the $t$th TA after $[\tilde{k}_{t},\tilde{l}_{t}]$-step MD-CDDS in AFDM. Therefore, we have the final $[\tilde{k}_{t},\tilde{l}_{t}]$-step MD-CDDS precoding matrix of AFDM as  
\begin{align}
	\mathbf{C}_{\text{MD-CDDS}}^{[\tilde{k}_{t},\tilde{l}_{t}],\text{AFDM}}=\boldsymbol{\Pi}_{N}^{\tilde{m}_{t}} \mathbf{P}_{N}^{[\tilde{m}_{t}]},
	\label{eq25.05.12.7}
\end{align}
which involves two procedures of first performing phase compensation and then performing cyclic DAFT-domain shift. Table \ref{tabel2} summarizes the precoding matrices of $[\tilde{k}_{t},\tilde{l}_{t}]$-step TD-CDDS and MD-CDDS in OTFS and AFDM systems.

\begin{table}[t]
	\renewcommand\arraystretch{1.7}
	\centering
	\caption{$[\tilde{k}_{t},\tilde{l}_{t}]$-Step TD-CDDS and MD-CDDS Precoding Matrices of OTFS and AFDM}
	\label{tabel2}
	\begin{tabular}{|c|c|}
		\hline
		TD-CDDS (Eq. \ref{eq2.14.1}) & $\qquad \mathbf{C}_{\text{TD-CDDS}}^{[\tilde{k}_{t},\tilde{l}_{t}]} = \boldsymbol{\Delta}_{N}^{\tilde{k}} \boldsymbol{\Pi}_{N}^{\tilde{l}} \qquad $\\ 
		\hline
		MD-CDDS-OTFS (\ref{eq25.05.12.6}) & $\mathbf{C}_{\text{MD-CDDS}}^{[\tilde{k}_{t},\tilde{l}_{t}],\text{OTFS}}=(\mathbf{I}_{L}
		\otimes \boldsymbol{\Pi}_{K}^{\tilde{k}})\boldsymbol{\Pi}_{KL}^{K*\tilde{l}_{t}}$ \\
		\hline
		MD-CDDS-AFDM (\ref{eq25.05.12.7}) & $ \qquad \mathbf{C}_{\text{MD-CDDS}}^{[\tilde{k}_{t},\tilde{l}_{t}],\text{AFDM}}=\boldsymbol{\Pi}_{N}^{\tilde{m}_{t}} \mathbf{P}_{N}^{[\tilde{m}_{t}]}\qquad$ \\			
		\hline
	\end{tabular}
\end{table}

\begin{remark} 
	\textup{Comparing (\ref{eq2.14.4}) and (\ref{eq2.16.2}) with (\ref{eq2.12.4}) and  (\ref{eq2.12.6}) correspondingly, we can observe that MD-CDDS can also be considered as shifting the delay and Doppler shifts of all the original propagation paths simultaneously by $\tilde{l}_{t}$ and $\tilde{k}_{t}$, respectively\footnote{Since Equation (\ref{eq2.12.6}) and the derivations in (\ref{eq25.05.12.10})-(\ref{eq25.05.12.7}) hold for all values of $c_{1}$ in AFDM, the proposed MD-CDDS scheme is also feasible for AFDM with a CPP that cannot be simplified to a CP.}.}
\end{remark}

Furthermore, substituting $\mathbf{x}, \mathbf{x}^{[\tilde{k}_{2},\tilde{l}_{2}]},\dots, \mathbf{x}^{[\tilde{k}_{N_{t}},\tilde{l}_{N_{t}}]}$ into $\mathbf{x}_{1}, \mathbf{x}_{2},\dots, \mathbf{x}_{N_{t}}$ in (\ref{eq100}), respectively, where $\mathbf{x}^{[\tilde{k}_{t},\tilde{l}_{t}]}=	\mathbf{C}^{[\tilde{k}_{t},\tilde{l}_{t}]}\mathbf{x}$, $\mathbf{C}^{[\tilde{k}_{t},\tilde{l}_{t}]}$ is the $[\tilde{k}_{t},\tilde{l}_{t}]$-step MD-CDDS precoding matrix of the adopted waveform ($\mathbf{C}_{\text{MD-CDDS}}^{[\tilde{k}_{t},\tilde{l}_{t}],\text{AFDM}}$ and $\mathbf{C}_{\text{MD-CDDS}}^{[\tilde{k}_{t},\tilde{l}_{t}],\text{OTFS}}$ for AFDM and OTFS, respectively),  we have the received modulation-domain signal at the $r$th RA as
\begin{align}
	\mathbf{y}_{r} &=\mathbf{H}_{r,1} \mathbf{x}+\mathbf{H}_{r,2} \mathbf{x}^{[\tilde{k}_{2},\tilde{l}_{2}]}+\cdots+\mathbf{H}_{r,N_{t}} \mathbf{x}^{[\tilde{k}_{N_{t}},\tilde{l}_{N_{t}}]}\notag \\
	&=\mathbf{H}_{r,1} \mathbf{x}+
	\underbrace{\mathbf{H}_{r,2} \mathbf{C}^{[\tilde{k}_{2},\tilde{l}_{2}]}}_{\mathbf{\bar{H}}_{r,2}}\mathbf{x}
	+\cdots+\underbrace{\mathbf{H}_{r,N_{t}} \mathbf{C}^{[\tilde{k}_{N_{t}},\tilde{l}_{N_{t}}]}}_{\mathbf{\bar{H}}_{r,N_{t}}}\mathbf{x} \notag \\
	&=\mathbf{\bar{H}}_{r}\mathbf{x},
	\label{eq25.05.16.1}
\end{align}
where $\mathbf{\bar{H}}_{r,t} = \mathbf{H}_{r,t} \mathbf{C}^{[\tilde{k}_{t},\tilde{l}_{t}]}$ is the effective modulation-domain channel matrix between the $r$th RA and the $t$th TA after $[\tilde{k}_{t},\tilde{l}_{t}]$-step MD-CDDS, $\mathbf{\bar{H}}_{r} = (\mathbf{H}_{r,1}+\mathbf{\bar{H}}_{r,2}+\dots+\mathbf{\bar{H}}_{r,N_{t}})$. This indicates that the MIMO system developed in (\ref{eq100}) can also be viewed as an equivalent SIMO system after MC-CDDS, where $\mathbf{\bar{H}}_{r}$ serves as the equivalent modulation-domain channel matrix between the equivalent single TA and the $r$th RA.

 \begin{figure}[tbp]
	\centering
	\includegraphics[width=0.47\textwidth,height=0.21\textwidth]{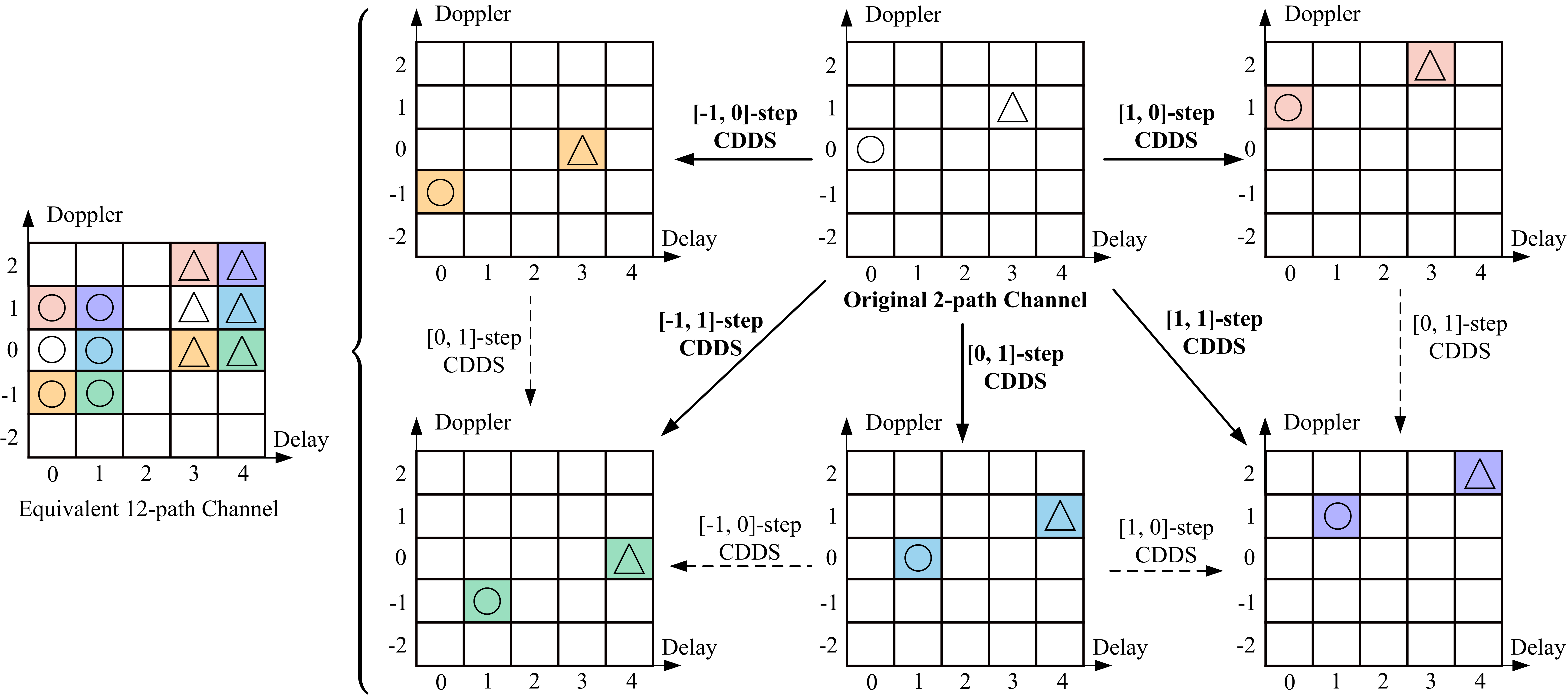}
	\caption{An example of the equivalent DD-domain channel after CDDS with six TAs.}
	\label{3-4}
\end{figure}

\subsection{Graphic Interpretation of CDDS}
We further provide a graphic interpretation of the principle CDDS in the sequel, as shown in Fig. \ref{3-4} (shown at the top of the next page). Assume that there are two paths in the original DSC between the transmitter and receiver, as indicated by “$\bigcirc $" and “$\bigtriangleup$", respectively, and the transmitter is equipped with six TAs. Then, the effective channels that after \{[-1, 0], [1, 0],  [-1, 1], [0, 1], [1, 1]\}-step CDDS are provided, respectively, which are associated with $\mathbf{\tilde{\bar{H}}}_{r,t}$ from (\ref{eq24.01.09.2}) in the time domain and  $\mathbf{\bar{H}}_{r,t}$ from (\ref{eq25.05.16.1}) in the modulation domain, where the corresponding effective paths are highlighted with different colors. Based on that, the equivalent 12-path channel
of the equivalent SIMO system can be obtained, which is associated with $\mathbf{\tilde{\bar{H}}}_{r}$ from (\ref{eq25.05.12.4}) in the time domain and  $\mathbf{\bar{H}}_{r}$ from (\ref{eq25.05.16.1}) in the modulation domain. Notably, as will be illustrated in Sec.~\ref{Sec4.2}, this equivalent single-TA characteristic in CDDS will significantly reduce channel estimation overhead compared to conventional MIMO-AFDM and MIMO-OTFS systems that require estimating a total of $N_{t}N_{r}$ DSCs across all pairs of TA and RA. Moreover, it is worth mentioning that, although the effective maximum delay is enlarged, the required length of CP remains unchanged, given that the delay shifts are performed cyclically in both time and modulation domains.

\subsection{Comparison between TD-CDDS and MD-CDDS}
From the viewpoint of the ultimate effect, the difference between TD-CDDS and MD-CDDS lies in the extra complex exponentials in the effective channel gains, i.e., $\bar{h}_{i}^{[r,t]}$ for TD-CDDS in (\ref{eq25.05.12.3}), $\bar{h}_{i, \text{OTFS}}^{[r,t]}$ for MD-CDDS-OTFS in (\ref{eq2.14.4}), and $\bar{h}_{i,\text{AFDM}}^{[r,t]}$ for MD-CDDS-AFDM in (\ref{eq2.16.2}). Moreover, while TD-CDDS is suitable for all waveforms, where all waveforms share the same TD-CDDS precoding matrix, the MD-CDDS precoding matrix should be designed elaborately according to the modulation-domain IOR of the specific waveform\footnote{Although this paper mainly focuses on AFDM and OTFS, both CDDS schemes can be directly extended to other waveforms that explore multipath for diversity, such as Zak-OTFS, ODDM, and orthogonal chirp division multiplexing (OCDM).}. Furthermore, TD-CDDS only requires one modulation operation, while MD-CDDS requires $N_{t}$ times modulation, which means  TD-CDDS enjoys lower complexity. However, since the CDDS is carried out before modulation in MD-CDDS, it is more feasible for MD-CDDS to combine with other precoding techniques to achieve joint precoding without introducing additional operation overhead to the transmitter.

Besides, it is worth emphasizing that the TD-CDDS and MD-CDDS matrices of OTFS and AFDM, namely $\mathbf{C}_{\text{TD-CDDS}}^{[\tilde{k}_{t},\tilde{l}_{t}]}$, $	\mathbf{C}_{\text{MD-CDDS}}^{[\tilde{k}_{t},\tilde{l}_{t}],\text{OTFS}}$, and $\mathbf{C}_{\text{MD-CDDS}}^{[\tilde{k}_{t},\tilde{l}_{t}],\text{AFDM}}$, are all sparse phase-compensated permutation matrices, which have no relevance to the DD profile of the real-time DSCs. It means that the TD-CDDS and MD-CDDS matrices can be calculated only once in advance with relatively low complexity at the transmitter, despite the ever-changing channels, and the implementation of CDDS only involves dot product and cyclic shift. Moreover, the equivalent single-TA property of CDDS makes it feasible to directly apply low-complexity signal detectors proposed for SISO systems to CDDS systems \cite{bb25.5.9.1}, which is of extreme significance from the perspective of practical implementation.

\section{Performance Analysis of CDDS}
\label{Sec4}
In this section, we evaluate the performance of the two proposed CDDS schemes.

\subsection{Transmit Diversity}
We first analyze the transmit diversity order of MIMO-AFDM and MIMO-OTFS systems with CDDS. Defining the DD profile of the channel before CDDS as $\mathbb{P} = \{(k_{1},l_{1}),\cdots,(k_{P},l_{P})\}$, then the effective DD profile after  $[\tilde{k}_{t},\tilde{l}_{t}]$-step CDDS can be denoted as $\mathbb{P}^{[\tilde{k}_{t},\tilde{l}_{t}]}$, and the effective DD profile of the equivalent SIMO system can be represented as $\mathbb{\tilde{P}}=\mathbb{P}\cup \mathbb{P}^{[\tilde{k}_{2},\tilde{l}_{2}]}\cup \cdots \cup\mathbb{P}^{[\tilde{k}_{N_{t}},\tilde{l}_{N_{t}}]}$. Then, we have the following proposition.

\begin{proposition}  \label{propos1}
	Optimal transmit diversity gain  quivalent to the number of transmit antennas $N_{t}$ can be achieved by MIMO-AFDM and MIMO-OTFS with CDDS in doubly selective channels if
	\begin{align}
		\mathbb{P}\cap  \mathbb{P}^{[\tilde{k}_{2},\tilde{l}_{2}]}\cap  \cdots \cap \mathbb{P}^{[\tilde{k}_{N_{t}},\tilde{l}_{N_{t}}]}=\varnothing \ \text{and} \ N_{t}P\leq N.
		\label{eq25.05.13.3}
	\end{align}	
\end{proposition} 
\begin{proof}
	See Appendix \ref{APP1}.
\end{proof}
Insightfully, (\ref{eq25.05.13.3}) indicates that to acquire optimal transmit diversity, one should carefully choose the CDDS steps at all TAs so that there is no overlap among the $N_{t}$ effective DD profiles that are associated with the $N_{t}$ TAs, as the example shown in Fig. \ref{3-4}\footnote{It was shown in \cite{23.11.18.2, bb26.2.10.1} that applying channel coding can potentially improve the ability of OTFS systems to acquire optimal diversity order. Therefore, CDDS-based MIMO-OTFS can effectively achieve optimal transmit diversity gain of $N_{t}$ with channel coding.}. Therefore, we refer to (\ref{eq25.05.13.3}) as \textbf{\emph{path non-overlap}} condition.

\begin{figure*}[htbp]
	\centering
	\includegraphics[width=0.88\textwidth,height=0.48\textwidth]{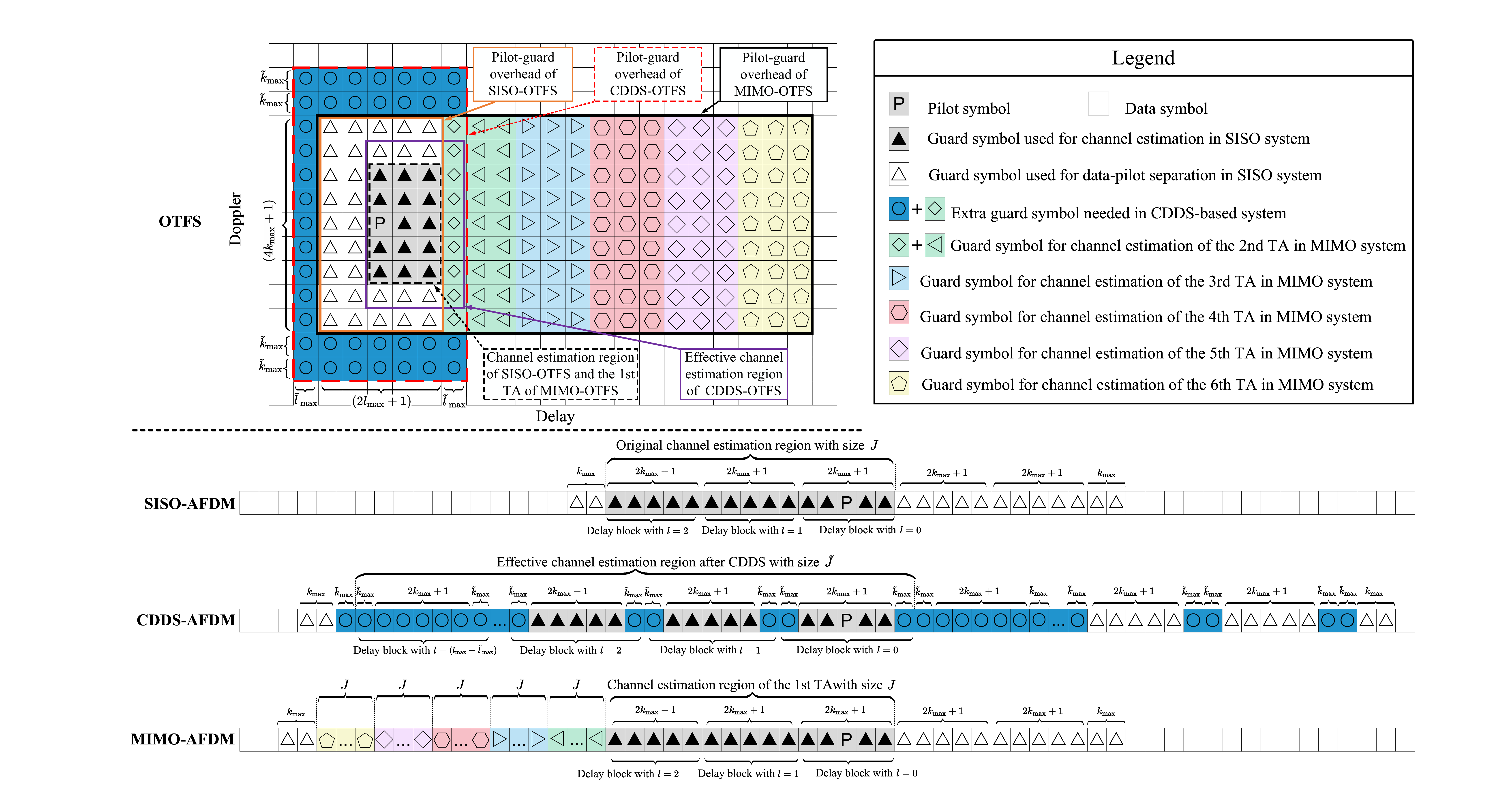}
	\caption{Data-pilot-guard symbol arrangement in AFDM and OTFS with SISO, MIMO, and CDDS configurations (six TAs). }
	\label{5-0}
\end{figure*}

\begin{table*}[ht]
	\renewcommand\arraystretch{1.3}
	\centering
	\caption{Channel Estimation Overhead of AFDM and OTFS Systems}
	\label{tabel7-31-1}
	\begin{tabular}{|c|c|c|}
		\hline
		\diagbox{Configuration}{Waveform}   & AFDM & OTFS\\ 
		\hline
		SISO \cite{bb23.1.3.4, bb25.5.20.1} & $2(l_{\max}+1)\big(2k_{\max }+1\big)-1$ & $(2l_{\max}+1)(4k_{\max }+1)$\\
		\hline
		MIMO \cite{bb23.2.10.7, bb25.5.20.1} & $(N_{t}+1)(l_{\max}+1)\big(2k_{\max }+1\big)-1$ & $\big(N_{t}(l_{\max}+1)+l_{\max}\big)\big(4k_{\max }+1\big)$ \\	
		\hline
				CDD \cite{bb23.2.10.10} &$2(l_{\max}+N_{t})(2k_{\max}+1)-1$
		& $\big( 2(l_{\max}+N_{t}-1)+1\big)(4k_{\max }+1)$ \\
		\hline
		DoDD &$2(l_{\max}+1)(2k_{\max}+N_{t}+1)-1$
		& $(2l_{\max}+1)(4k_{\max}+ 2N_{t}+1)$ \\
		\hline
		CDDS &$2(l_{\max}+\tilde{l}_{\max}+1)\big(2(k_{\max}+\tilde{k}_{\max})+1\big)-1$
		& $( 2(l_{\max}+\tilde{l}_{\max})+1)(4(k_{\max }+\tilde{k}_{\max})+1)$ \\
		\hline
	\end{tabular}
\end{table*}

\subsection{EPA Channel Estimation Overhead}
\label{Sec4.2}
We next analyze the channel estimation overhead of CDDS-based AFDM and OTFS systems with embedded pilot-aided (EPA) channel estimation scheme \cite{bb23.1.3.4, bb23.2.10.7, bb25.5.20.1}. Fig.~\ref{5-0} (shown on the next page) shows the data-pilot-guard symbol arrangement of AFDM and OTFS with SISO, MIMO, and CDDS configurations. We can observe that the channel estimation overhead of CDDS-based AFDM and OTFS systems are much lower than the counterparts of the conventional MIMO-AFDM and MIMO-OTFS systems, whose data-pilot-guard symbol arrangement is adopted by most of the conventional transmit diversity techniques, e.g., Alamouti STC \cite{bb23.2.10.8,bb23.2.10.9}. This can be attributed to the fact that the CDDS-based system is equivalent to a SIMO system, which sufficiently explores the channel sparsity and hence only requires a few extra guard symbols compared to the SISO system. In contrast, the conventional MIMO systems separately estimate the channel parameters between all pairs of TA and RA, which require a significant increase of guard symbols.  
 
Furthermore, Table \ref{tabel7-31-1} (shown on the next page) quantifies the pilot-guard overhead of AFDM and OTFS with SISO, MIMO, and CDDS configurations for channel estimation, where 
 \begin{align}
 	\tilde{k}_{\max}\triangleq\max\{k_{\max}, k_{2,\max},\dots,k_{N_{t},\max}\}-k_{\max}
 \end{align}
 and $\tilde{l}_{\max}\triangleq \max\{\tilde{l}_{2},\dots,\tilde{l}_{N_{t}}\}$
 represent the maximum extra Doppler shift and maximum cyclic-delay shift, respectively, $k_{t,\max}$ and $\tilde{l}_{t}$ are defined in (\ref{eq25.05.12.9}) and (\ref{eq2.14.1}), respectively, ${k}_{\text{space}}$ is set to be $\tilde{k}_{\max}$,  $J=(l_{\max}+1)\big(2k_{\max }+1\big)$, and $\tilde{J}=(l_{\max}+\tilde{l}_{\max}+1)\big(2(k_{\max}+\tilde{k}_{\max})+1\big)$.Notably, the pilot-guard overhead of AFDM and OTFS systems with conventional MIMO, CDD, or DoDD configuration grows linearly with $N_{t}$. In contrast, the pilot-guard overhead of the proposed CDDS scheme exhibits a similar form as the SISO case and thus increases more slowly with the increase of $N_{t}$. This sharp distinction will be illustrated in Fig.~\ref{5-2} in Sec.~\ref{Sec5.2}.

\subsection{CDDS-Step Selection Criterion}
\label{sec4-4}
Building on the above insights, one optimal CDDS-step selection strategy is to minimize the channel estimation overhead while guaranteeing the optimal transmit diversity. This can be expressed as
\begin{align}
	&\min_{\tilde{k}_{2},\dots,\tilde{k}_{N_{t}},\tilde{l}_{2},\dots,\tilde{l}_{N_{t}}} O_{\text{CDDS-AFDM}} \ \text{or} \ O_{\text{CDDS-OTFS}} \notag \\ &  \qquad \quad \text{s.t.} \qquad  \qquad \qquad  \quad \ (\ref{eq25.05.13.3}),
	\label{eq25.05.22.1}
\end{align}
where notations $O_{\text{CDDS-AFDM}}$ and $O_{\text{CDDS-OTFS}}$ stand for the channel estimation overhead of CDDS-AFDM and CDDS-OTFS, respectively.

In practice, one can flexibly and empirically adjust the CDDS step  $\{\tilde{k}_{2},\dots,\tilde{k}_{N_{t}},\tilde{l}_{2},\dots,\tilde{l}_{N_{t}}\}$  through modulation-domain channel prediction to achieve or approximate (\ref{eq25.05.22.1}) given that the modulation-domain representations of DSCs in AFDM and OTFS are typically sparse and relatively stable. Furthermore, it is worth noting that a low channel estimation overhead calls for small $\tilde{k}_{\max}$ and $\tilde{l}_{\max}$, while the path non-overlap condition in (\ref{eq25.05.13.3}) prefers large  $\tilde{k}_{\max}$ and $\tilde{l}_{\max}$ to accommodate the optimal transmit diversity requirement. Therefore, there is a natural tradeoff between the channel estimation overhead and the achievable transmit diversity gain in CDDS-AFDM and CDDS-OTFS systems.

\subsection{Advantages over Conventional Schemes}
\label{sec4-5}
Compared to the classic Alamouti scheme \cite{bb23.2.10.8,bb23.2.10.9}, the proposed CDDS not only enjoys a much lower channel estimation overhead but also has a more flexible and simple TA and RA management. Specifically, for general Alamouti-based MIMO-OTFS and MIMO-AFDM systems, each data symbol is transmitted in $\eta$ frames ($\eta$ is the transmit diversity order achieved), where the precoding operation of each TA varies when the number of TAs changes, and the receiver should perform maximal-ratio combining correspondingly. This means that high latency is inevitable, and the receiver should precisely know the number of TAs. In contrast, the proposed CDDS scheme is conducted within a single frame, which achieves much lower latency. Moreover, increasing or decreasing the number of TAs in CDDS only introduces specific changes in the relative TAs and the receiver only needs to know $(l_{\max}+\tilde{l}_{\max})$ and $ (k_{\max}+\tilde{k}_{\max})$ to determine the channel estimation region, which is much simpler than the Alamouti scheme.

Furthermore, compared to the conventional CDD and DDoD schemes,  CDDS possesses a higher degree of flexibility thanks to an extra degree of freedom to perform cyclic shift, which enables it to achieve a more compelling tradeoff between transmit diversity and channel estimation overhead. Besides, it was shown in \cite{bb25.6.18.1} that combining Alamouti and CDD in OFDM can achieve better performance than pure CDD, which might be a promising extension of CDDS.

\section{Simulation Results}
\label{Sec5}
In this section, we verify the effectiveness of the two proposed CDDS schemes through simulation. We first adopt the maximum likelihood (ML) optimal detector with a small frame size and perfect CSI to evaluate the theoretical achievable transmitted diversity gain. Then, the widely-used message passing (MP) detector proposed in \cite{bb25.5.9.1} with large frame size and imperfect CSI are adopted to examine the robustness of the proposed CDDS scheme in MIMO-AFDM and MIMO-OTFS systems. In particular, $a \times b$ denotes a MIMO configuration of $a$ TAs and $b$ RAs. If not otherwise specified, rectangular pulse shaping is applied in both AFDM and OTFS systems.
For each channel realization, the channel coefficient $h_{i}$ follows the distribution of $\mathcal{C N}(0, 1 / P)$.

\begin{figure}[tbp]
	\centering
	\includegraphics[width=0.40\textwidth,height=0.345\textwidth]{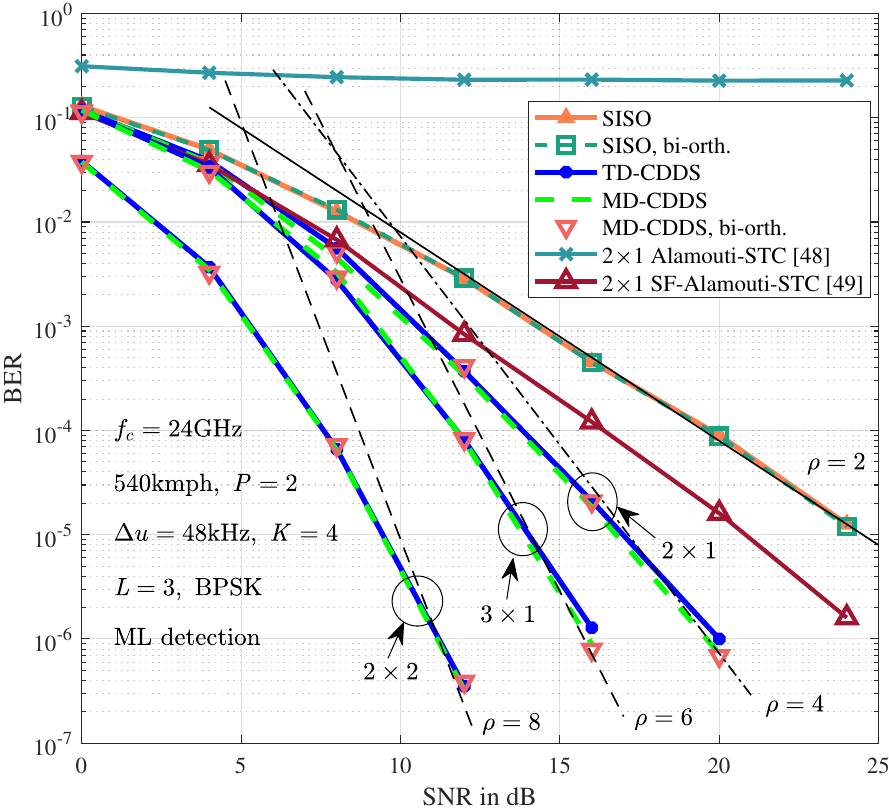}
	\caption{BER performance of OTFS systems with different transmit and receive antenna configurations and transmit diversity schemes, integer Doppler.}
	\label{4-1}
\end{figure}

\subsection{Perfect CSI}
We adopt carrier frequency $f_{c}=24$ GHz, number of subcarriers $N=12$, AFDM subcarrier spacing $\Delta f=12$ kHz, number of Doppler samples in OTFS frame $K=4$, number of delay samples in OTFS frame $L=3$, OTFS subcarrier spacing $\Delta u=48$ kHz to ensure the same TF resources are occupied by AFDM and OTFS. $P=2$ paths with DD profile of $\{(0,0), (-1,1)\}$ is applied, corresponding to a maximum Doppler shift of 12 kHz and a maximum UE speed of 540 kmph\footnote{To guarantee optimal transmit diversity gain in small frame size, we adopt a fixed DD profile channel and apply \{$[1,2]$, $[1,0]$, $[0,2]$\}-step CDDS on the 2nd, the 3rd, and the 4th TA of MIMO-AFDM systems, respectively, and \{$[1,1]$, $[-1,0]$\}-step CDDS on the 2nd and the 3rd TA of MIMO-OTFS systems, respectively, to satisfy the path non-overlap condition.}. BPSK and ML optimal detector are used.

Fig. \ref{4-1} shows the BER performance of OTFS systems with different antenna configurations and transmit diversity schemes. The asymptotic lines with slopes $\rho$ of 2, 4, 6, and 8 are provided for the convenience of comparison. We can observe that the diversity orders of OTFS systems with SISO, $2 \times 1$ CDDS, $3 \times 1$ CDDS, $3 \times 1$ CDDS, and $2 \times 2$ CDDS are approximately 2, 4, 6, and 8, respectively, which are associated with Proposition \ref{propos1}. Moreover, while the BER curves of TD-CDDS-OTFS and MD-CDDS-OTFS exhibit the same diversity order, their BER values differ slightly due to the difference in the effective channel gains, as shown in TD-CDDS (\ref{eq24.01.09.2}) and MD-CDDS (\ref{eq2.14.4}). Meanwhile, the MD-CDDS-OTFS with rectangular and bi-orthogonal pulse shaping exhibit the same BER performance, which indicates that the proposed CDDS schemes have robustness to the shaping pulse adopted. Furthermore, the $2 \times 1$ CDDS-based OTFS outperforms the $2 \times 1$ Alamouti-OTFS in \cite{bb23.2.10.8} significantly.  This can be attributed to the fact that Alamouti-OTFS requires the DSC to remain unchanged in two successive frames, a condition that cannot be satisfied in high-mobility scenarios. In this case, strong inter-frame interference that is significantly greater than the AWGN occurs in Alamouti-OTFS, causing very poor BER performance that does not improve with increasing SNR.   Meanwhile, the $2 \times 1$ CDDS-based OTFS also outperforms the $2 \times 1$ single-frame (SF)-Alamouti-OTFS in \cite{bb23.2.10.8}. This is because the equivalent Alamouti codeword in SF-Alamouti-OTFS is not strictly block orthogonal, which may affect its ability to provide optimal transmit diversity gain.

\begin{figure}[tbp]
	\centering
	\includegraphics[width=0.40\textwidth,height=0.34\textwidth]{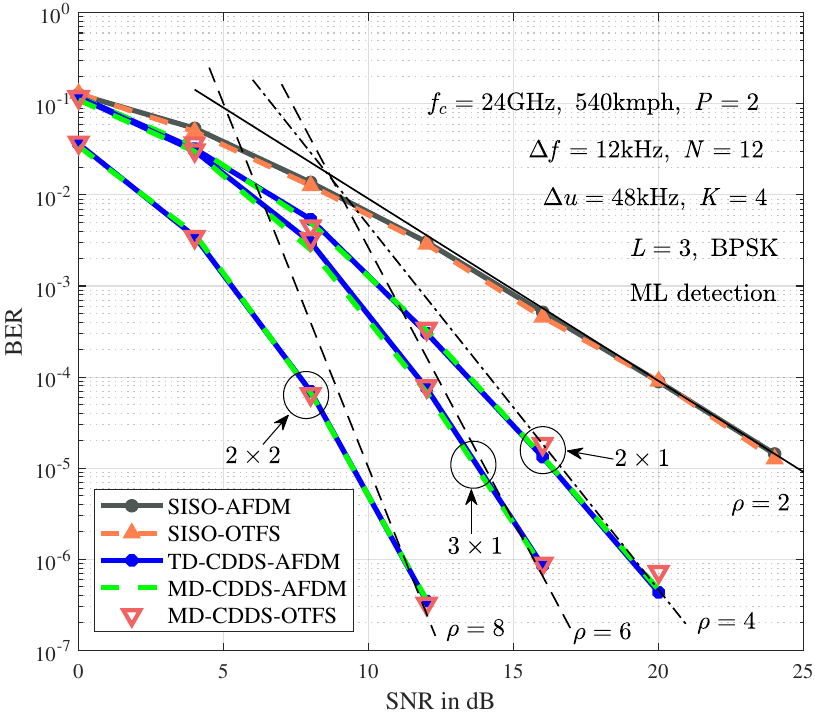}
	\caption{BER performance of AFDM systems with different transmit and receive antenna configurations, integer Doppler.}
	\label{4-1-2}
\end{figure}

\begin{figure}[tbp]
	\centering
	\includegraphics[width=0.40\textwidth,height=0.345\textwidth]{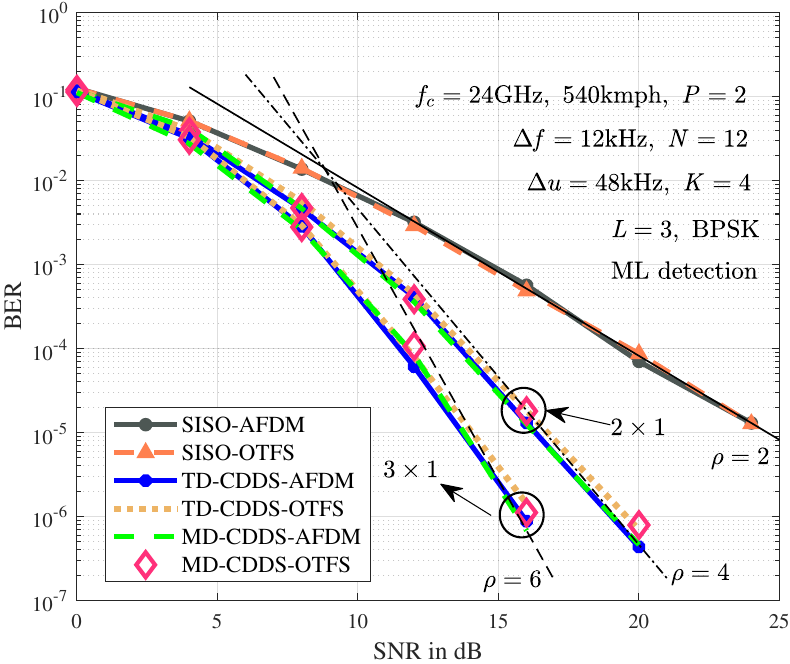}
	\caption{BER performance of OTFS and AFDM systems with  fractional Doppler, DD profile of $\{(0.2,0), (-0.9,1)\}$ is applied.}
	\label{4-2}
\end{figure}

Fig. \ref{4-1-2} shows the BER performance of AFDM systems with different antenna configurations. OTFS with SISO and MD-CDDS configurations are provided for ease of comparison. We can observe that the diversity orders of AFDM systems with SISO, $2 \times 1$ CDDS, $3 \times 1$ CDDS, $3 \times 1$ CDDS, and $2 \times 2$ CDDS are 2, 4, 6, and 8, respectively, which confirms Proposition \ref{propos1}. Moreover, we can notice from Fig.\ref{4-1} and Fig.~\ref{4-1-2} that both MD-CDDS and TD-CDDS offer optimal transmit diversity gain for MIMO-AFDM and MIMO-OTFS. This is because the extra complex exponentials of the effective channel gains induced by TD-CDDS in (\ref{eq25.05.12.3}) and MD-CDDS in (\ref{eq2.14.4}) and (\ref{eq2.16.2}) will not change the statistical characters of the original channel gains, as also demonstrated in Appendix \ref{APP1}. Additionally, Fig.~\ref{4-2} shows that similar conclusions can be drawn in DSCs with fractional Doppler shifts.

\begin{figure}[tbp]
	\centering
	\includegraphics[width=0.40\textwidth,height=0.34\textwidth]{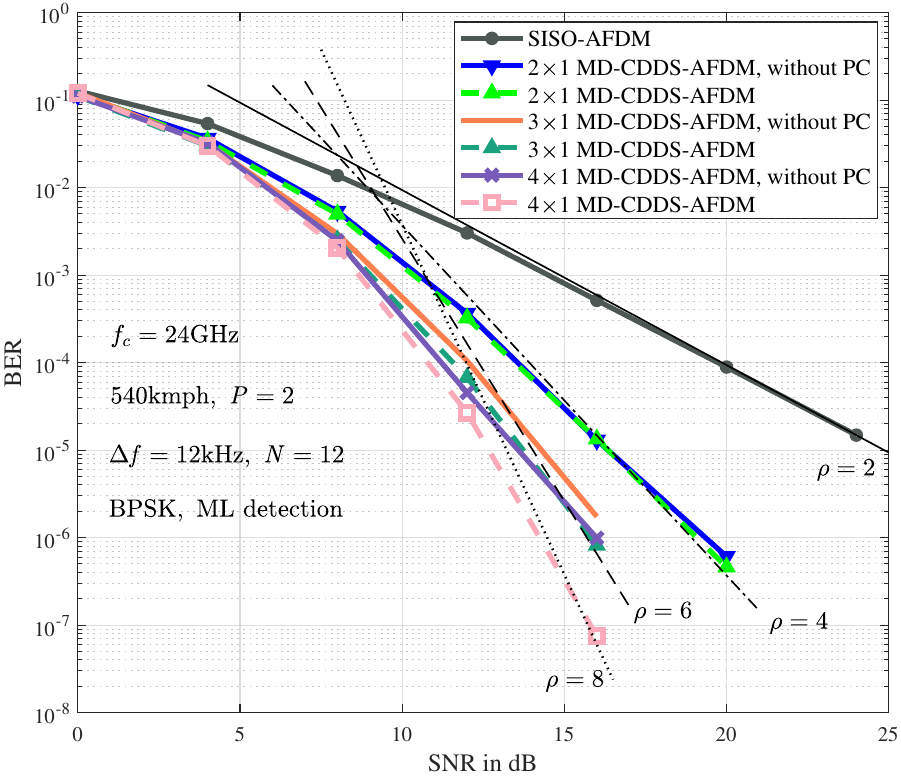}
	\caption{BER performance of MD-CDDS-AFDM with and without phase compensation, integer Doppler.}
	\label{4-11}
\end{figure}

We then investigate the influence of phase compensation operation in MD-CDDS-AFDM systems with Fig. \ref{4-11}. We can observe that MD-CDDS-AFDM systems with phase compensation deliver better BER performance than those without phase compensation. This is because, as the IOR of the MD-CDDS with phase compensation shown in (\ref{eq2.16.2}), the extra complex exponential $\mathcal{A}(l_{i}, \tilde{m}_{t})$ in the effective channel gain has no relevance to the index of DAFT-domain transmit symbol $\bar{m}$, which means all transmit DAFT-domain symbols undergo the same fading and hence enables robust optimal transmit diversity gain. By contrast,  the additional complex exponential of $\mathcal{E}(\bar{m}, \tilde{m}_{t})$ in the effective channel gain of MD-CDDS without phase compensation, as shown in (\ref{eq25.05.15.2}), is related to $\bar{m}$, which means each DAFT-domain symbol undergoes different fading and thereby induces performance loss.

\begin{figure}[tbp]
	\centering
	\includegraphics[width=0.40\textwidth,height=0.345\textwidth]{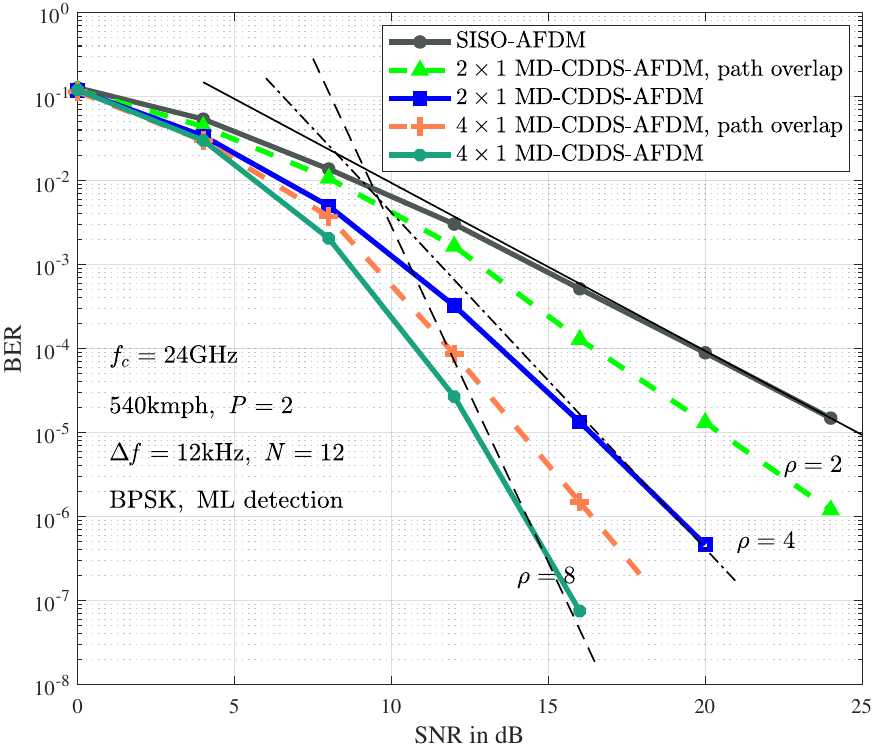}
	\caption{BER performance of CDDS-AFDM with different CDDS steps.}
	\label{4-10}
\end{figure}

We next show the BER performance of CDDS-AFDM systems with different CDDS-step selections. We can notice that the $[-1,1]$-step MD-CDDS-AFDM achieves lower diversity gain than the $[1,2]$-step MD-CDDS-AFDM, given that the former (with $\operatorname{card}(\mathbb{\tilde{P}})=3$) does not satisfy the path non-overlap condition as the latter (with $\operatorname{card}(\mathbb{\tilde{P}})=4$) does. Similar results can be obtained from the $4\times 1$ MD-CDDS-AFDM systems ($\operatorname{card}(\mathbb{\tilde{P}})=6$ and $\operatorname{card}(\mathbb{\tilde{P}})=8$ for the case with and without path overlap, respectively), which indicate that the CDDS-step selections greatly influence the achievable transmit diversity gain. As discussed at the end of Sec.~\ref{sec4-5}, it is difficult for the conventional CDD and DoDD schemes in \cite{bb23.2.10.10} to efficiently satisfy the path non-overlap condition to provide optimal transmit diversity given their fixed shifts limited to either the delay or Doppler domain. Moreover, it is worth mentioning that the CDDS-AFDM system with $[1,2]$-step CDDS typically requires more guard symbols than that with $[-1,1]$-step CDDS in EPA channel estimation for its larger maximum cyclic-delay shift, revealing the fundamental tradeoff between the channel estimation overhead and the achievable transmit diversity gain in CDDS-based systems discussed in Sec. \ref{sec4-4}.

\subsection{Imperfect CSI}
\label{Sec5.2}
We next investigate the effectiveness of the proposed CDDS scheme with estimated CSI and practical MP detector. We adopt $f_{c}=24$ GHz, $N=1024$, $\Delta f=6$ kHz, $K=32$, $L=32$, $\Delta u=192$ kHz to ensure the same TF resources are occupied by AFDM and OTFS systems. $P=2$ paths, where the delay indices are chosen randomly according to the uniform distribution among $[0, l_{\max}]$ and the Doppler indices are generated by using Jakes’ formula, i.e., $k_{i}=k_{\max } \cos \left(\theta_{i}\right)$, where $\theta_{i}$ is uniformly distributed over $[-\pi, \pi]$. $l_{\max}=4$ and $k_{\max}=3$ are set, corresponding to a maximum Doppler shift of 18 kHz and a maximum UE speed of 810 kmph. The EPA diagonal reconstruction (EPA-DR) scheme proposed in \cite{bb23.2.10.7} and the EPA scheme proposed in \cite{bb25.5.20.1} are applied to perform channel estimation in AFDM and OTFS systems, respectively. The pilot signal for SNR is represented as SNRp, while the data signal for SNR is represented as SNRd.

\begin{figure}[tbp]
	\centering
	\includegraphics[width=0.40\textwidth,height=0.33\textwidth]{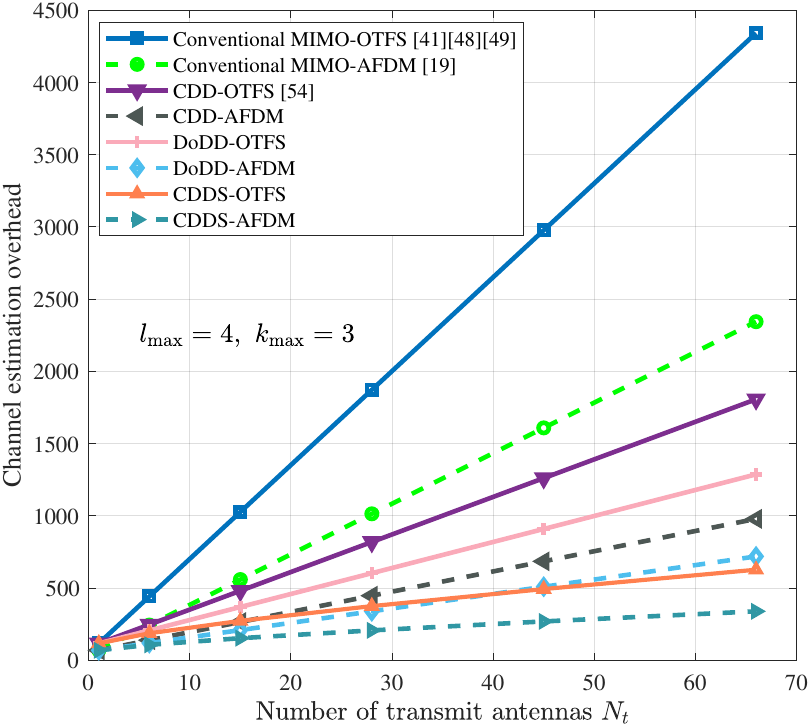}
	\caption{Comparison of channel estimation overhead between AFDM and OTFS with different transmit diversity schemes.}
	\label{5-2}
\end{figure}

Fig. \ref{5-2} compares the channel estimation overhead of AFDM and OTFS with conventional MIMO, CDD, DoDD, and the proposed CDDS configurations, which are obtained according to Table \ref{tabel7-31-1}. We can observe that the channel estimation overheads of AFDM and OTFS systems with the proposed CDDS scheme are significantly lower than those of the conventional transmit diversity schemes, and the gap between them becomes larger as the increase of the number of TAs.  This aligns well with the discussion in Sec. \ref{Sec4.2} and implies that CDDS is particularly suitable for large-scale MIMO to perform transmit diversity precoding. Moreover, we can also notice that AFDM maintains its compelling advantage of fewer channel estimation overhead over OTFS in the CDDS setting, allowing it to support a larger scale of MIMO.

\begin{figure}[tbp]
	\centering
	\includegraphics[width=0.41\textwidth,height=0.345\textwidth]{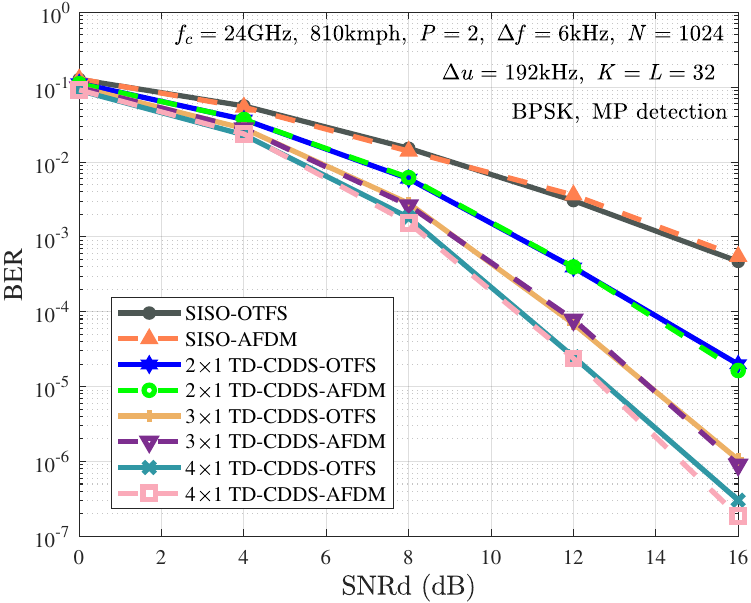}
	\caption{BER performance of TD-CDDS-based AFDM and OTFS systems with estimated CSI, SNRp = 40 dB, MP detector, integer Doppler.}
	\label{4-3}
\end{figure}

\begin{figure}[tbp]
	\centering
	\includegraphics[width=0.41\textwidth,height=0.345\textwidth]{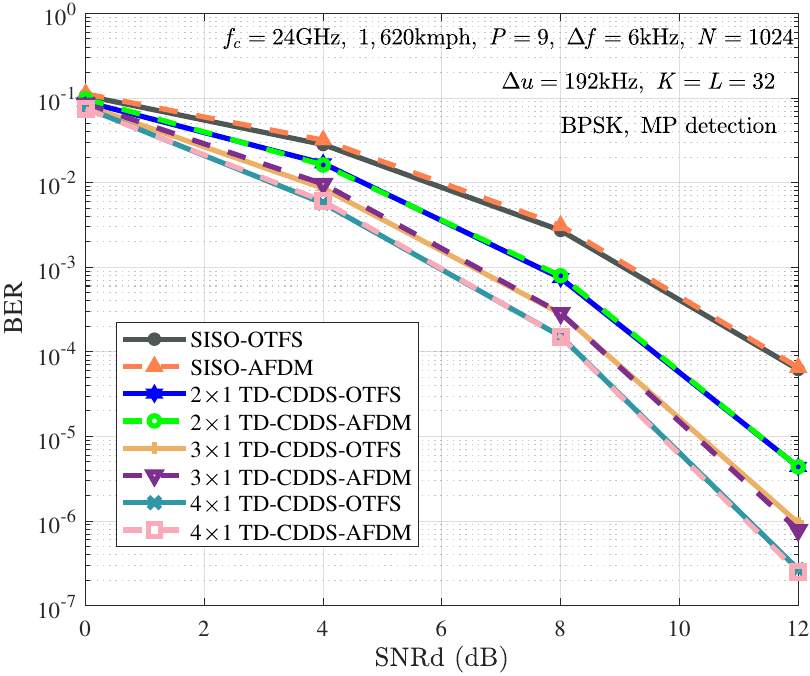}
	\caption{BER performance of TD-CDDS-based AFDM and OTFS systems with estimated CSI, SNRp = 40 dB, MP detector, integer Doppler, EVA channel model \cite{bb26.2.16.1}.}
	\label{4-4}
\end{figure}

We proceed to validate the performance of the proposed CDDS scheme under a non-optimal detector. Fig.~\ref{4-3} shows the BER performance of CDDS-based AFDM and OTFS systems with estimated CSI, SNRp = 40 dB. Only the TD-CDDS case is applied, given that the MD-CDDS case is expected to deliver the same performance. We can observe that the BER performance of AFDM and OTFS improves eminently as the number of TAs increases. This is because the MP detector can sufficiently explore the inherent diversity of AFDM and OTFS systems, and CDDS can offer them optimal transmit diversity gain, enabling ultra-reliable communications in practical scenarios. Moreover, we further validate the effectiveness of the proposed CDDS scheme under realistic Extended Vehicular A (EAV)
channel model with $P=9$, $l_{\max}=15$, $k_{\max }=6$, and a relative power profile of [0.0, -1.5, -1.4, -3.6, -0.6, -9.1, -7.0, -12.0, -16.9] (dB) in Fig.~\ref{4-4} \cite{bb26.2.16.1}, which demonstrates the robustness of the proposed scheme.

\section{Conclusion}
\label{Sec6}
In this work, a novel transmit diversity technique, named cyclic delay-Doppler shift, for MIMO-AFDM and MIMO-OTFS is presented. We demonstrate the detailed implementation of CDDS on MIMO-AFDM and MIMO-OTFS in the time domain and modulation domain by deriving the corresponding CDDS precoding matrices, which are nothing but sparse phase-compensated permutation matrices. It is shown that the proposed CDDS scheme can offer MIMO-AFDM and MIMO-OTFS with optimal transmit diversity gain when the CDDS step is chosen to satisfy the path non-overlap condition, which gives rise to a fundamental tradeoff between the channel estimation overhead and the achievable transmit diversity gain. Compared to the conventional Alamouti scheme, the proposed scheme is more flexible and easier to implement, showing great potential in enabling ultra-reliable and low-latency communications in high-mobility scenarios. Future work will focus on integrating channel prediction into CDDS to fully explore its potential.

{\appendices
	\section{Proof of Proposition \ref{propos1}}
	\label{APP1}
	We provide the proof of the AFDM case, where the OTFS case can be obtained similarly by following the same procedures. Consider the equivalent SISO system between the transmitter with CDDS and the $r$th RA, then the received DAFT-domain vector can be denoted as 
	\begin{align}
		\mathbf{y}_{r} = \sum_{t=1}^{N_{t}}\sum_{i=1}^{P}\bar{h}_{i}^{[t]} \mathbf{\bar{H}}_{i}^{[t]} \mathbf{x}+\mathbf{w}_{r}=\boldsymbol{\Phi}(\mathbf{x}) \mathbf{h}+\mathbf{w}_{r},
		\label{eq25.05.17.1}
	\end{align}
	where $\bar{h}_{i}^{[t]}$ and $\mathbf{\bar{H}}_{i}^{[t]}$ are the effective channel gain and effective modulation-domain subchannel matrix between the $t$th TA  and the $r$th RA after CDDS (for TD-CDDS-AFDM, $\bar{h}_{i}^{[t]}=\bar{h}_{i}^{[r,t]}$ and $\mathbf{\bar{H}}_{i}^{[t]}=\mathbf{A}\mathbf{\tilde{\bar{{H}}}}_{i}^{[r,t]}\mathbf{A}^{H}$, while for MD-CDDS-AFDM, $\bar{h}_{i}^{[t]}=\bar{h}_{i,\text{AFDM}}^{[r,t]}$ and $\mathbf{\bar{H}}_{i}^{[t]}=\mathbf{H}_{i}^{\text{AFDM}}\mathbf{C}_{\text{MD-CDDS}}^{[\tilde{k}_{t},\tilde{l}_{t}],\text{AFDM}}$), and
	\begin{align}
		&\mathbf{\Phi}(\mathbf{x}) =  \notag \\
		&\left[
		\mathbf{\bar{H}}_{1}^{[1]} \mathbf{x},  \ldots,   \mathbf{\bar{H}}_{P}^{[1]} \mathbf{x},  \dots,  \mathbf{\bar{H}}_{1}^{[N_{t}]} \mathbf{x},  \ldots,  \mathbf{\bar{H}}_{P}^{[N_{t}]} \mathbf{x}
		\right]\in \mathbb{C}^{N \times N_{t}P},
		\label{eq25.05.17.2}
	\end{align}
	channel gain vector 
	\begin{align}
		\mathbf{h} = \left[\bar{h}_{1}^{[1]}, \ldots, \bar{h}_{P}^{[1]},\dots, \bar{h}_{1}^{[N_{t}]}, \dots, \bar{h}_{P}^{[N_{t}]}\right]^{T} \in \mathbb{C}^{N_{t}P\times 1}.
		\label{eq25.05.17.3}
	\end{align}
	Assume $h_{i}^{[r,t]}$ to be i.i.d and distributed as $\mathcal{C N}(0,1 / N_{t}P)$, then $\bar{h}_{i}^{[t]}$ is also i.i.d and distributed as $\mathcal{C N}(0,1 / N_{t}P)$ given that both $\bar{h}_{i}^{[r,t]}$ and $\bar{h}_{i,\text{AFDM}}^{[r,t]}$ are obtained by multiplying $\bar{h}_{i}^{[r,t]}$ with a complex exponential. Moreover, we normalize the transmit vector $\mathbf{x}$ so that the SNR at each receive antenna is $\frac{1}{N_{0}}$. 
	
	Let $\mathbf{x}$ and $\mathbf{\hat{x}}$ be two transmit symbol matrices. Assuming perfect CSI and maximum likelihood (ML) detection at the receiver, the probability of transmitting the symbol matrix $\mathbf{x}$ and deciding in favor of $\mathbf{\hat{x}}$ at the receiver is the conditional PEP between $\mathbf{x}$ and $\mathbf{\hat{x}}$, which can be expressed as 
	\begin{equation}
		P\left(\mathbf{x} \rightarrow \mathbf{\hat{x}} \mid \mathbf{\tilde{h}}, \mathbf{x}\right)=Q\left(\sqrt{\frac{\left\|\left(
				\mathbf{\Phi}(\mathbf{x})-\mathbf{\Phi}(\mathbf{\hat{x}})\right)\mathbf{h}\right\|^{2}}{2 N_{0}}}\right).
	\end{equation}
	where $\left\|\left(
	\mathbf{\Phi}(\mathbf{x})-\mathbf{\Phi}(\mathbf{\hat{x}})\right)\mathbf{h}\right\|^{2}
	=\mathbf{h}^{H} \mathbf{\Omega}\mathbf{h}$,
	\begin{align}
		\mathbf{\Omega} = \left(
		\mathbf{\Phi}(\mathbf{x})-\mathbf{\Phi}(\mathbf{\hat{x}})\right)^{H}\left(
		\mathbf{\Phi}(\mathbf{x})-\mathbf{\Phi}(\mathbf{\hat{x}})\right)  = \mathbf{U} \boldsymbol{\hat{\Lambda}} \mathbf{U}^{H},
		\label{eq25.05.17.4}
	\end{align}
	$\mathbf{U}$ is a unitary matrix and $\boldsymbol{\hat{\Lambda}}=\operatorname{diag}\left(\lambda_{1}^{2}, \dots, \lambda_{N_{t}P}^{2}\right)$ with $\lambda_{i}$ being the $i$th singular value of the difference matrix $\boldsymbol{\delta}=\mathbf{\Phi}(\mathbf{x})-\mathbf{\Phi}(\mathbf{\hat{x}})$.  Consequently, we have
	\begin{align}
		\mathbb{E}_{\mathbf{h}}\{\left\|\left(
		\mathbf{\Phi}(\mathbf{x})-\mathbf{\Phi}(\mathbf{\hat{x}})\right)\mathbf{h}\right\|^{2}\}=\sum_{i=1}^{\theta}\frac{\lambda_{i}^{2}}{N_{t}P},
		\label{eq25.05.17.7}
	\end{align}
	where $\theta$ is the rank of $\boldsymbol{\delta}$.
	Considering $
	Q(x) \cong  \frac{1}{12} e^{-x^{2} / 2}+\frac{1}{4} e^{-2 x^{2} / 3}$ and $e^{-x}\leq \frac{1}{1+x}$
	, an upper bound of the unconditional PEP can be obtained as
	\begin{align}
		P\left(\mathbf{x} \rightarrow \mathbf{\hat{x}}\right)\leq \frac{1}{12}\prod_{i=1}^{\theta}\frac{1}{1+\frac{\lambda_{i}^{2}}{4N_{0}N_{t}P}}
		+\frac{1}{4}\prod_{i=1}^{\theta}\frac{1}{1+\frac{\lambda_{i}^{2}}{3N_{0}N_{t}P}}.
		\label{eq25.05.17.8}
	\end{align}
	At high SNR regime, (\ref{eq25.05.17.8}) can be further simplified as
	\begin{align}
		P\left(\mathbf{x} \rightarrow \mathbf{\hat{x}}\right)\leq \left(\prod_{i=1}^{\theta} \frac{\lambda_{i}^{2}}{N_{t}P} \right)^{-1}\left(\frac{4^{\theta}+3^{\theta+1}}{12}\right)\left(\frac{1}{N_{0}}\right)^{\theta}.
		\label{eq25.05.17.9}
	\end{align}
We can observe from the right term of (\ref{eq25.05.17.9}) that the diversity order of the equivalent SISO system is $ \theta$. According to Remark \ref{Remark1}, the diversity order of the SISO-AFDM system is $P$, indicating that $\theta$ is equal to the number of separable paths \cite{bb23.1.3.4}. Insightfully, when (\ref{eq25.05.13.3}) is satisfied, we have
$\operatorname{card}(\mathbb{\tilde{P}})=N_{t}P$, hence $\theta=N_{t}P$, i.e., optimal transmit diversity gain $N_{t}$ can be acquired. This completes the proof of Proposition \ref{propos1}.
}

\end{document}